\PassOptionsToPackage{unicode}{hyperref}
\PassOptionsToPackage{hyphens}{url}
\PassOptionsToPackage{dvipsnames,svgnames,x11names}{xcolor}
\documentclass[12pt]{article}
\usepackage{amsmath,amssymb}
\usepackage{iftex}
\ifPDFTeX
  \usepackage[T1]{fontenc}
  \usepackage[utf8]{inputenc}
  \usepackage{textcomp} % provide euro and other symbols
\else % if luatex or xetex
  \usepackage{unicode-math}
  \usepackage[toc]{appendix}
  \defaultfontfeatures{Scale=MatchLowercase}
  \defaultfontfeatures[\rmfamily]{Ligatures=TeX,Scale=1}
\fi
\usepackage{lmodern}
\ifPDFTeX\else  
\fi
\IfFileExists{upquote.sty}{\usepackage{upquote}}{}
\IfFileExists{microtype.sty}{% use microtype if available
  \usepackage[]{microtype}
  \UseMicrotypeSet[protrusion]{basicmath} % disable protrusion for tt fonts
}{}
\makeatletter
\@ifundefined{KOMAClassName}{% if non-KOMA class
  \IfFileExists{parskip.sty}{%
    \usepackage{parskip}
  }{% else
    \setlength{\parindent}{0pt}
    \setlength{\parskip}{6pt plus 2pt minus 1pt}}
}{% if KOMA class
  \KOMAoptions{parskip=half}}
\makeatother
\usepackage{xcolor}
\makeatletter
\ifx\paragraph\undefined\else
  \let\oldparagraph\paragraph
  \renewcommand{\paragraph}{
    \@ifstar
      \xxxParagraphStar
      \xxxParagraphNoStar
  }
  \newcommand{\xxxParagraphStar}[1]{\oldparagraph*{#1}\mbox{}}
  \newcommand{\xxxParagraphNoStar}[1]{\oldparagraph{#1}\mbox{}}
\fi
\ifx\subparagraph\undefined\else
  \let\oldsubparagraph\subparagraph
  \renewcommand{\subparagraph}{
    \@ifstar
      \xxxSubParagraphStar
      \xxxSubParagraphNoStar
  }
  \newcommand{\xxxSubParagraphStar}[1]{\oldsubparagraph*{#1}\mbox{}}
  \newcommand{\xxxSubParagraphNoStar}[1]{\oldsubparagraph{#1}\mbox{}}
\fi
\makeatother

\usepackage{longtable,booktabs,array}
\usepackage{calc} % for calculating minipage widths
\usepackage{etoolbox}
\makeatletter
\patchcmd\longtable{\par}{\if@noskipsec\mbox{}\fi\par}{}{}
\makeatother
\IfFileExists{footnotehyper.sty}{\usepackage{footnotehyper}}{\usepackage{footnote}}
\makesavenoteenv{longtable}
\usepackage{graphicx}
\makeatletter
\def\maxwidth{\ifdim\Gin@nat@width>\linewidth\linewidth\else\Gin@nat@width\fi}
\def\maxheight{\ifdim\Gin@nat@height>\textheight\textheight\else\Gin@nat@height\fi}
\makeatother
\setkeys{Gin}{width=\maxwidth,height=\maxheight,keepaspectratio}
\makeatletter
\def\fps@figure{htbp}
\makeatother

\makeatletter
\@ifpackageloaded{caption}{}{\usepackage{caption}}
\AtBeginDocument{%
\ifdefined\contentsname
  \renewcommand*\contentsname{Table of contents}
\else
  \newcommand\contentsname{Table of contents}
\fi
\ifdefined\listfigurename
  \renewcommand*\listfigurename{List of Figures}
\else
  \newcommand\listfigurename{List of Figures}
\fi
\ifdefined\listtablename
  \renewcommand*\listtablename{List of Tables}
\else
  \newcommand\listtablename{List of Tables}
\fi
\ifdefined\figurename
  \renewcommand*\figurename{Figure}
\else
  \newcommand\figurename{Figure}
\fi
\ifdefined\tablename
  \renewcommand*\tablename{Table}
\else
  \newcommand\tablename{Table}
\fi
}
\@ifpackageloaded{float}{}{\usepackage{float}}
\floatstyle{ruled}
\@ifundefined{c@chapter}{\newfloat{codelisting}{h}{lop}}{\newfloat{codelisting}{h}{lop}[chapter]}
\floatname{codelisting}{Listing}

\makeatother
\makeatletter
\@ifpackageloaded{caption}{}{\usepackage{caption}}
\@ifpackageloaded{subcaption}{}{\usepackage{subcaption}}
\makeatother

\ifLuaTeX
  \usepackage{selnolig}  % disable illegal ligatures
\fi
\usepackage{natbib}
\usepackage{bookmark}

\IfFileExists{xurl.sty}{\usepackage{xurl}}{} % add URL line breaks if available
\hypersetup{
  pdftitle={Title},
  pdfauthor={Author 1; Author 2},
  pdfkeywords={3 to 6 keywords, that do not appear in the title},
  colorlinks=true,
  linkcolor={blue},
  filecolor={Maroon},
  citecolor={Blue},
  urlcolor={Blue},
  pdfcreator={LaTeX via pandoc}}

\newtheorem{theorem}{Theorem}[section]
\newtheorem{corollary}{Corollary}[section]

\newtheorem{assumption}{A.}
\usepackage{algorithm}
\usepackage{algpseudocode}

\newcommand{\anon}{1}

\begin{document}

\def\spacingset#1{\renewcommand{\baselinestretch}%
{#1}\small\normalsize} \spacingset{1}

%%%%%%%%%%%%%%%%%%%%%%%%%%%%%%%%%%%%%%%%%%%%%%%%%%%%%%%%%%%%%%%%%%%%%%%%%%%%%%

\if1\anon
{
  \title{\bf Towards Optimal Estimators for\\ 
  Randomized Control Trials}
  \author{Harsh Parikh\hspace{.2cm}\\
    SCOT, Amazon \\ 
    Department of Biostatistics, Yale University \and
    Gabriel Levin-Konigsberg \\
    SCOT, Amazon \and
    Nilesh Tripuraneni\\ Google \and Dhruv Madeka\\ Anthropic \and
    Michael I. Jordan\\UC Berkeley \and
    Dean Foster \\
    SCOT, Amazon \and
    Dominique Perrault-Joncas \\
    SCOT, Amazon \and
    Alexander Volfovsky \\
    SCOT, Amazon \\
    Department of Statistical Science, Duke University \\
    }
  \maketitle
} \fi

\if0\anon
{
  \bigskip
  \bigskip
  \bigskip
  \begin{center}
    {\LARGE\bf Title}
\end{center}
  \medskip
} \fi

\bigskip
\begin{abstract}
    Randomized controlled trials (RCTs) are fundamental tools for causal inference across technology companies, pharmaceutical research, and federal agencies. While the standard difference-in-means estimator provides unbiased treatment effect estimates, it often lacks precision, particularly when treatment effects are heterogeneous or outcomes exhibit heavy-tailed distributions. Although numerous precision-enhancing methods exist---from covariate adjustment techniques to variance reduction strategies---recent research demonstrates that no single estimator performs optimally across all datasets. Rather than seeking the best estimator for individual RCTs, which risks compromising scientific validity through convenient selection, we propose a principled framework for identifying optimal estimators within families of RCTs based on specific analytical goals. Our approach uses sample splitting to estimate the distribution of evaluation metrics (e.g., mean squared error, regret) across RCT families, enabling systematic comparisons between estimators while maintaining asymptotic guarantees. We demonstrate this framework using a sample of Amazon's Supply Chain Optimization Technology trials and the Strengthening Democracy Challenge dataset (25 interventions). Results reveal that optimal estimators vary significantly by analytical objective: weighted least squares performs best for inference goals, while difference-in-means minimizes regret for decision-making contexts. This work provides actionable guidance for estimator selection while preserving methodological rigor across diverse research applications.

\end{abstract}

\noindent%
\vfill

\newpage
\spacingset{1.8} % DON'T change the spacing!

	\section{Introduction}

Randomized controlled trials (RCTs) have become the gold standard for causal inference across diverse scientific domains.
% from technology companies like Amazon, Meta, and Google optimizing user experiences, to pharmaceutical giants such as Eli Lilly and Pfizer evaluating drug efficacy, to federal agencies including the FDA and NIH informing policy decisions. 
% Often the outcome is same across all these trial. 
The appeal of RCTs lies in their exceptional internal validity: treatment randomization effectively eliminates unobserved confounding, ensuring unbiased estimation of causal effects is possible. The knowledge from these experiments serves dual purposes---advancing scientific understanding of underlying causal mechanisms and informing critical downstream decisions about intervention deployment. Each experiment requires an estimation strategy that translates the observed outcomes into an estimate of the treatment effect, and currently there is no single strategy for selecting such an estimator for a class of related RCTs. 

A foundational estimator for RCTs is the the standard difference-in-means (DM) estimator, which simply contrasts the average outcomes among those assigned to treatment and control in the experiment. While this estimator is unbiased it can suffer from inadequate precision. This limitation becomes particularly acute when treatment effects exhibit substantial heterogeneity or when outcome distributions display heavy-tailed characteristics---scenarios increasingly common in modern applications. The precision challenge has motivated extensive methodological development along two primary axes: covariate adjustment strategies and variance reduction techniques. 
% First, 
% researchers have pursued 
% covariate adjustment strategies that leverage pretreatment variables associated with outcomes or treatment effects. 
The landscape of available methods has expanded dramatically in both directions due to advances in machine learning and causal inference, encompassing approaches from simple ordinary least squares to sophisticated random forests and neural networks to changes in the targeted estimand. 
% Second, variance reduction techniques have emerged that either directly control outcome distributions (such as winsorization) or modify the estimand itself (such as data-adaptive weighted average treatment effects).

This methodological abundance naturally raises a fundamental question: \textit{which estimator should researchers choose for their analysis?} Validation studies in causal inference have consistently demonstrated that no universal ``optimal'' estimator exists across all datasets. Selecting estimators on a case-by-case basis, however, introduces significant risks to scientific validity, as researchers may gravitationally select methods that support preferred narratives rather than those that provide the most reliable inference. We address this challenge by shifting focus from finding optimal estimators for individual RCTs to identifying superior approaches for entire families of related trials. 

Within each scientific domain, hundreds or thousands of experiments are conducted each year with each experiment measuring the impact of potentially different treatments on the \textit{same outcome}: e.g. technology companies like Amazon, Meta, and Google optimizing user experiences where the outcomes are click-through rates, pharmaceutical giants such as Eli Lilly and Pfizer evaluating drug efficacy in terms of all-cause-mortality and federal agencies including the FDA and NIH informing policy decisions --- where each experiment aims to evaluate the impact on the same outcome. 
% This perspective recognizes that many research contexts involve repeated experimentation on similar populations, interventions, or outcomes---whether through platform experimentation at technology companies, drug development pipelines at pharmaceutical firms, or policy evaluation programs at government agencies.

Our contribution is a general framework that characterizes optimal estimators for specified classes of RCTs, explicitly accounting for intended analytical goals and outcome characteristics. The framework employs sample splitting methodology to estimate evaluation metric distributions across trial families, ensuring valid asymptotic behavior while enabling systematic estimator comparisons. We demonstrate the framework's utility through two case studies: Amazon's Supply Chain Optimization Technology (SCOT) trials comprising a sample of 556 experiments RCTs focused on a common financial outcome metric, and the Strengthening Democracy Challenge dataset encompassing 25 interventions targeting political attitudes and behaviors \citep{voelkel2024mega}.
Our empirical findings confirm the absence of universally optimal estimators while providing actionable guidance for method selection. For the Amazon SCOT trials, we find that analytical objectives critically determine optimal choices. Similarly, analysis of the Democracy Challenge data reveals outcome-dependent optimal estimators, reinforcing the importance of context-specific method selection.

This work makes several key contributions to the statistical literature on causal inference. First, we provide a principled alternative to ad hoc estimator selection that preserves scientific integrity while accommodating methodological flexibility. Second, we demonstrate how evaluation criteria should align with analytical objectives, distinguishing between inference and decision-making goals. Third, we offer a practical framework that practitioners can readily implement across diverse experimental contexts. Finally, our empirical applications showcase the framework's value in both academic and business settings, highlighting its broad applicability across research domains.

Section~\ref{sec:literature} discusses about relevant literature. Section~\ref{sec:prelim} introduces the setup and notation as well as discusses some of the existing estimator for treatment effects. Section~\ref{sec:val} introduces our framework and method for characterizing optimal estimator. Section~\ref{sec:amazon} discusses our business case study on Amazon data and Section~\ref{sec:case_study} discusses our social science case study on SDC data.

        \section{Relevant Literature}
\label{sec:literature}
Randomized controlled trials have long been considered the gold standard for estimating causal effects across various disciplines \citep{imbens2015causal}. The fundamental appeal of RCTs lies in their ability to balance confounders across treatment conditions via treatment randomization, thereby enabling unbiased estimation of average treatment effects. The simplest approach is the difference-in-means estimator, which directly compares outcome averages between treatment and control groups \citep{athey2017econometrics}. While unbiased, researchers have increasingly recognized that precision can be substantially improved through more sophisticated methods without compromising internal validity \citep{Benkeser2020}.

\paragraph*{Covariate Adjustment in RCTs.} Covariate adjustment has become standard for improving precision in RCT analysis. Incorporating pre-treatment covariates through regression models can significantly reduce variance in treatment effect estimates \citep{freedman2008regression, lin2013agnostic}. Unlike in observational studies, the primary motivation here is variance reduction rather than addressing confounding \citep{tsiatis2008covariate}. \citet{yang2001efficiency} showed that ANCOVA estimators remain consistent even under model misspecification while providing efficiency gains. Recent work has explored high-dimensional regression adjustment using regularization methods such as LASSO \citep{bloniarz2016lasso} and ridge regression \citep{wager2016high}.

\paragraph*{Methods from Observational Studies Applied to RCTs.} Several methods developed for observational studies have been adapted to improve precision in RCTs. Inverse Probability Weighting (IPW) estimators account for chance imbalances in covariates \citep{hirano2003efficient}, while meta-learners including T-learners, S-learners, and X-learners model treatment effect heterogeneity using machine learning approaches \citep{kunzel2019metalearners}. Doubly robust estimators like Augmented IPW (AIPW) combine outcome regression with propensity score methods \citep{funk2011doubly}, with recent extensions to double/debiased machine learning \citep{chernozhukov2018double}.

\paragraph*{The Challenge of Method Validation and Selection.} With this proliferation of methods, the challenge has shifted from developing new estimators to selecting among existing alternatives. Traditional validation approaches for causal inference methods fall into three categories: face-validity tests that assess estimates against expert intuition, placebo or negative control tests that evaluate a method's ability to recover null effects under domain-driven constraints, and synthetic data studies with known data-generating mechanisms.

\paragraph*{Synthetic Data Approaches and Their Limitations.} Recent literature has explored sophisticated synthetic data generation to identify optimal estimators for particular contexts. Approaches using variational autoencoders (VAEs) and generative adversarial networks (GANs) have been proposed to simulate datasets with known causal structures \citep{athey2018gan, parikh2022validating}. Particularly notable is the work by \citet{parikh2022validating}, who developed a comprehensive framework for validating causal inference methods using VAE-based synthetic data generation, enabling systematic comparison across varying degrees of confounding and treatment effect heterogeneity.

Despite these advances, synthetic data approaches face important limitations that motivate our alternative framework. As noted by \citet{parikh2022validating} and \citet{gentzel2019case}, performance is highly sensitive to training procedures and hyperparameter choices of generative models. \citet{losch2021optimal} found significant divergence between estimator rankings on synthetic versus real-world data, suggesting that synthetic performance may not translate reliably to actual applications. Moreover, the lack of theoretical guarantees regarding synthetic-to-real performance relationships creates the risk of ``convenient estimator choice,'' where researchers might selectively report results from synthetic settings favorable to preferred methods.

\paragraph*{Gap in Current Literature.} While substantial progress has been made in developing precision-enhancing methods and validation approaches, a critical gap remains in providing principled guidance for estimator selection in practice. Current approaches typically focus on either developing new methods or evaluating methods in isolation, rather than providing systematic frameworks for choosing among competing approaches based on research context and analytical objectives. Our work addresses this gap by proposing a framework that leverages families of related RCTs to identify optimal estimators while accounting for both intended use cases (inference versus decision-making) and distributional properties of outcomes within specific research domains. Unlike synthetic approaches, our framework directly evaluates estimators on collections of real RCTs sharing common outcome measures, providing stronger assurance regarding cross-context generalizability while maintaining principled statistical foundations.
        \section{Background}
\label{sec:prelim}

\subsection{Setup \& Notations}
We consider a framework where both individual studies and their design characteristics are drawn from broader populations, allowing us to develop estimator selection principles that generalize across families of related experiments.

Let $\mathcal{O}$ denote a countably infinite superpopulation of units eligible for treatment interventions, such as all potential customers, products, or individuals in a given domain. We consider a finite set of potential treatments $\mathcal{T} = \{0, 1, \ldots, K-1\}$, where treatment $0$ typically represents a control or baseline condition, and the remaining treatments represent various interventions of interest.

Our analysis focuses on a collection of $N$ studies, indexed by the random variable $S \in \{1, 2, \ldots, N\}$. Each study $S = s$ is generated through the following two-stage sampling process:

\textit{Stage 1: Treatment Selection.} For study $s$, we first randomly sample two treatments $t_a^{(s)}$ and $t_b^{(s)}$ uniformly at random from $\mathcal{T} \times \mathcal{T}$. These represent the two treatment conditions to be compared in the study, where typically one serves as the active treatment and the other as a control or comparison condition.

\textit{Stage 2: Population Sampling.} Given the selected treatment pair $(t_a^{(s)}, t_b^{(s)})$, we recruit $m_s$ units from the superpopulation $\mathcal{O}$ through some sampling mechanism (which may be random or non-random) to form the study sample $\mathcal{O}_s \subseteq \mathcal{O}$. The sampling mechanism may depend on study-specific factors such as eligibility criteria, temporal constraints, or resource limitations.

Thus, each study $s$ is uniquely identified by the triple $(t_a^{(s)}, t_b^{(s)}, \mathcal{O}_s)$, encompassing both the treatments under comparison and the specific population recruited for that comparison.

The estimand of interest for study $s$ is the average treatment effect within the study population:
\begin{equation}
\tau^{(s)} = \mathbb{E}[Y_i(t_a^{(s)}) - Y_i(t_b^{(s)}) \mid S = s],
\end{equation}
where $Y_i(t)$ denotes the potential outcome for unit $i$ under treatment $t$, following the standard potential outcomes framework \citep{neyman1923application, rubin1974estimating}.

For each unit $i \in \mathcal{O}_s$ recruited into study $s$, we observe the following data:
\begin{itemize}
    \item \textbf{Pretreatment covariates}: $X_i \in \mathbb{R}^p$, representing baseline characteristics measured before treatment administration
    \item \textbf{Treatment assignment}: $T_i \in \{t_a^{(s)}, t_b^{(s)}\}$, the randomly assigned treatment condition
    \item \textbf{Outcome}: $Y_i \in \mathbb{R}$, measured at a fixed time point $k$ periods after treatment administration
\end{itemize}

The observed dataset for study $s$ is $\mathcal{D}_s = \{(Y_i, T_i, X_i) : i \in \mathcal{O}_s\}$. Our framework leverages the collection of datasets $\{\mathcal{D}_s\}_{s=1}^N$ from this family of related studies to identify optimal estimators for the average treatment effects $\{\tau^{(s)}\}_{s=1}^N$.

We make the following key assumptions that enable our analysis:

\begin{assumption}[Randomization within Studies]
\label{ass:randomization}
Within each study $s$, treatment assignment is randomized such that
$Y_i(t) \perp T_i \mid S = s,$ and $P(T_i = t \mid X_i) > 0$ $\forall t \in \{t_a^{(s)}, t_b^{(s)}\}$, $\forall i \in \mathcal{O}_s,$
\end{assumption}

\begin{assumption}[Covariate-Conditional Study Independence]
\label{ass:study_independence}
Conditional on covariates, potential outcomes are independent of study membership: $Y_i(t) \perp S \mid X_i, \quad \forall t \in \mathcal{T}, \quad \forall i \in \mathcal{O}.$
\end{assumption}

Assumption~\ref{ass:study_independence} ensures that studies differ only in their covariate distributions, not in the underlying causal mechanisms.

\begin{assumption}[Non-overlapping Study Samples]
\label{ass:non_overlap}
Each unit in the superpopulation participates in at most one study:
$\mathcal{O}_s \cap \mathcal{O}_{s'} = \emptyset, \quad \forall s \neq s'.$
\end{assumption}

Under Assumption \ref{ass:randomization}, the average treatment effect in study $s$ can be identified as:
\begin{equation}
\tau^{(s)} = \mathbb{E}[Y_i \mid T_i = t_a^{(s)}, S = s] - \mathbb{E}[Y_i \mid T_i = t_b^{(s)}, S = s].
\end{equation}

Assumption \ref{ass:study_independence} is crucial for our framework as it allows us to meaningfully compare estimator performance across studies---it ensures that any differences in estimator performance across studies can be attributed to differences in covariate distributions and sample compositions rather than fundamental differences in causal mechanisms. Assumption \ref{ass:non_overlap} simplifies the statistical analysis by ensuring independence across study samples and avoiding complications from repeated measurements of the same units.

This setup naturally accommodates scenarios where studies share common treatments (when the same treatment pairs are drawn multiple times) as well as cases where studies evaluate entirely different treatment comparisons, providing a flexible framework for analyzing families of related experiments.
	\section{Framework for Comparing Estimators}
\label{sec:val}

This section presents our methodology for rigorously assessing and comparing treatment effect estimators across families of randomized controlled trials. Our approach leverages the collection of studies to estimate the distribution of estimator performance, enabling principled comparisons while avoiding overfitting to any single dataset.

\subsection{Performance Metrics for Different Analytical Goals}

Let $\hat{\tau}^{(s)}_\alpha$ denote the estimated treatment effect when applying estimator $\alpha$ to study $s$, where the specific treatments $(t_a^{(s)}, t_b^{(s)})$ and sample $\mathcal{O}_s$ are determined by the study design. We assess estimator performance using metrics $\Delta(\cdot, \cdot)$ that quantify the discrepancy between estimates and true treatment effects, with the choice of metric depending on the intended use of the estimates.

\textbf{Inference-Focused Applications.} For applications where the primary goal is accurate parameter estimation, we use mean squared error:
\begin{equation}
\Delta_{\text{MSE}}(\hat{\tau}^{(s)}_\alpha, \tau^{(s)}) = (\hat{\tau}^{(s)}_\alpha - \tau^{(s)})^2.
\end{equation}
This metric captures both bias and variance, providing a comprehensive measure of estimator precision for statistical inference.

\textbf{Decision-Making Applications.} For applications where estimates inform treatment rollout decisions, we focus on the economic consequences of estimation errors. We consider a simple decision rule where treatment is rolled out if $\hat{\tau}^{(s)}_\alpha > c \cdot \text{SE}(\hat{\tau}^{(s)}_\alpha)$ for some critical threshold $c \geq 0$. The regret from suboptimal decisions is:
\begin{equation}
\Delta_{\text{Regret}}(\hat{\tau}^{(s)}_\alpha, \tau^{(s)}, c) = |\mathcal{O}_s| \times \left( \tau^{(s)} \times \left(\mathbf{1}[\tau^{(s)} > 0] - \mathbf{1}[\hat{\tau}^{(s)}_\alpha > c \cdot \text{SE}(\hat{\tau}^{(s)}_\alpha)]\right) \right),
\end{equation}
where $|\mathcal{O}_s|$ represents the scale of potential implementation. This metric quantifies the opportunity cost when beneficial treatments are not rolled out (false negatives) and the harm when ineffective treatments are implemented (false positives).

% \textbf{Power.} 

% $P(\hat{\tau}^{(s)}_\alpha < c \cdot \text{SE}(\hat{\tau}^{(s)}_\alpha) \mid \tau > 0)$

\subsection{Framework for Estimator Comparison}

For a given performance metric $\Delta$, we define the expected performance of estimator $\alpha$ in study $s$ as:
\begin{equation}
\epsilon^{(s)}_\alpha = \mathbb{E}_{\mathcal{O}_s}\left[\Delta(\hat{\tau}^{(s)}_\alpha, \tau^{(s)}) \mid S = s\right],
\end{equation}
where the expectation is taken over the randomness in sample composition $\mathcal{O}_s$ and treatment assignment within the study.

The overall performance of estimator $\alpha$ across the family of studies is:
\begin{equation}
M_\alpha = \mathbb{E}_S\left[\epsilon^{(s)}_\alpha\right] = \mathbb{E}_S\left[\mathbb{E}_{\mathcal{O}_s}\left[\Delta(\hat{\tau}^{(s)}_\alpha, \tau^{(s)}) \mid S = s\right]\right].
\end{equation}

We say that estimator $\alpha$ dominates estimator $\beta$ if $M_\alpha < M_\beta$. To test this comparison, we define the test statistic:
\begin{equation}
\theta(\alpha,\beta) = M_\beta - M_\alpha,
\end{equation}
where $\theta(\alpha,\beta) > 0$ indicates that $\alpha$ outperforms $\beta$.

\subsection{Sample-Splitting Estimation Procedure}

Since true treatment effects $\{\tau^{(s)}\}_{s=1}^N$ are unknown, we cannot directly compute the performance metrics. We propose a sample-splitting approach that provides valid estimates of $\theta(\alpha,\beta)$ while avoiding overfitting.

\textbf{Cross-Fitting Protocol.} For each study $s$, we randomly partition the sample $\mathcal{O}_s$ into $K$ non-overlapping subsets of approximately equal size: $\{\mathcal{O}_{s,k}\}_{k=1}^K$ such that $\cup_{k=1}^K \mathcal{O}_{s,k} = \mathcal{O}_s$ and $\mathcal{O}_{s,i} \cap \mathcal{O}_{s,j} = \emptyset$ for $i \neq j$.

For each fold $k$, we:
\begin{enumerate}
\item Use subsample $\mathcal{O}_{s,k}$ to compute a difference-in-means estimate: $\hat{\tau}^{(s)}_{\text{DM},k} = \bar{Y}_{s,k}^{(1)} - \bar{Y}_{s,k}^{(0)}$, where $\bar{Y}_{s,k}^{(j)}$ is the sample mean outcome for treatment $j$ in fold $k$.
\item Use the remaining data $\mathcal{O}_{s,-k} = \mathcal{O}_s \setminus \mathcal{O}_{s,k}$ to compute the estimator of interest: $\hat{\tau}^{(s)}_{\alpha,-k}$.
\end{enumerate}

The cross-fitted performance estimate for estimator $\alpha$ in study $s$ is:
\begin{equation}
\hat{\epsilon}^{(s)}_\alpha = \frac{1}{K} \sum_{k=1}^K \Delta\left(\hat{\tau}^{(s)}_{\alpha,-k}, \hat{\tau}^{(s)}_{\text{DM},k}\right).
\end{equation}

This procedure leverages the fact that under Assumption \ref{ass:randomization}, the difference-in-means estimator is unbiased: $\mathbb{E}[\hat{\tau}^{(s)}_{\text{DM},k}] = \tau^{(s)}$. The cross-fitting ensures that the "ground truth" estimate and the estimator being evaluated use independent data, preventing overfitting.

The estimated overall performance is $\hat{M}_\alpha = \frac{1}{N} \sum_{s=1}^N \hat{\epsilon}^{(s)}_\alpha$, and the test statistic becomes:
\begin{equation}
\hat{\theta}(\alpha,\beta) = \hat{M}_\beta - \hat{M}_\alpha = \frac{1}{N} \sum_{s=1}^N \left(\hat{\epsilon}^{(s)}_\beta - \hat{\epsilon}^{(s)}_\alpha\right).
\end{equation}

\begin{theorem}[Unbiasedness of Cross-Fitted MSE-based Test Statistic]
\label{thm:unbiased_estimator}
Given assumption A.1-3, the cross-fitted MSE statistic is unbiased: $\mathbb{E}\left[\hat{\theta}_{MSE}(\alpha, \beta) - \theta_{MSE}(\alpha, \beta) \right] = 0.$
\end{theorem}

\begin{theorem}[Bias Bounds for Regret-Based Test Statistic]
\label{thm:regret_bias_bounds}
For the regret performance metric with threshold $c \geq 0$, the bias of the cross-fitted test statistic is bounded as: $|\hat{\theta}_{\text{regret}}(\alpha, \beta) - \theta_{\text{regret}}(\alpha, \beta)| \leq \frac{1}{N} \sum_{s=1}^N \left(B^{(s)}_\alpha + B^{(s)}_\beta\right),$
where for any estimator $\gamma$ and study $s$:
\begin{equation*}
B^{(s)}_\gamma = |\mathcal{O}_s| \times \mathbb{E}\left[|\tau^{(s)}| \times |\mathbf{1}[\tau^{(s)} > 0] - \mathbf{1}[\hat{\tau}^{(s)}_{\text{DM},k} > 0]| + |\tau^{(s)} - \hat{\tau}^{(s)}_{\text{DM},k}| \mid S = s\right].
\end{equation*}
\end{theorem}

\begin{theorem}[Expected Bias of Regret-Based Test Statistic]
\label{thm:regret_expected_bias}
For the regret performance metric, the expected bias of the cross-fitted test statistic is:
\begin{equation}
\mathbb{E}[\hat{\theta}_{\text{regret}}(\alpha, \beta) - \theta_{\text{regret}}(\alpha, \beta)] = \frac{1}{N} \sum_{s=1}^N \left(B^{(s)}_\beta - B^{(s)}_\alpha\right),
\end{equation}
where for any estimator $\gamma$ and study $s$:
\begin{align}
B^{(s)}_\gamma &= |\mathcal{O}_s| \times \bigg[\mathbb{E}\left[(\hat{\tau}^{(s)}_{\text{DM},k} - \tau^{(s)}) \times \left(\mathbf{1}[\hat{\tau}^{(s)}_{\text{DM},k} > 0] - \mathbf{1}[\hat{\tau}^{(s)}_{\gamma,-k} > 0]\right) \mid S = s\right] 
\\
&\quad + \tau^{(s)} \times \left(P(\hat{\tau}^{(s)}_{\text{DM},k} > 0 \mid S = s) - \mathbf{1}[\tau^{(s)} > 0]\right)\bigg].
\end{align}
Furthermore, the second term satisfies:
\begin{equation}
\tau^{(s)} \times \left(P(\hat{\tau}^{(s)}_{\text{DM},k} > 0 \mid S = s) - \mathbf{1}[\tau^{(s)} > 0]\right) \leq 0 \quad \forall s.
\end{equation}
\end{theorem}

\begin{corollary}[Asymptotic Unbiasedness]
Under regularity conditions ensuring that $\hat{\tau}^{(s)}_{\text{DM},k} \xrightarrow{p} \tau^{(s)}$ as $|\mathcal{O}_{s,k}| \to \infty$, we have
$B^{(s)}_\gamma \to 0 \text{ for all } s, \gamma,$
implying that $\hat{\theta}_{\text{regret}}(\alpha, \beta)$ is asymptotically unbiased for $\theta_{\text{regret}}(\alpha, \beta)$.
\end{corollary}
The bound $B^{(s)}_\gamma$ has an intuitive interpretation: (i) the first term captures the cost of sign errors when $\hat{\tau}^{(s)}_{\text{DM},k}$ and $\tau^{(s)}$ have opposite signs, and (ii) the second term captures the cost of magnitude errors in the DM estimator. 

\subsection{Statistical Inference}
Establishing asymptotic normality for $\hat{\theta}(\alpha,\beta)$ faces several complications: the asymptotic behavior depends critically on the smoothness of the performance metric $\Delta$ (our regret metric involves indicator functions that create discontinuities), our cross-fitting procedure introduces two-stage estimation complexity, and study heterogeneity in sample sizes and effect magnitudes can lead to non-identically distributed performance differences requiring specialized limit theorems. Additionally, multiple comparisons across estimators require family-wise error rate corrections that further complicate CLT-based inference.

\textbf{Advantages of Permutation-Based Inference.} Permutation tests elegantly sidestep these technical challenges by leveraging the fundamental exchangeability property under the null hypothesis $H_0: \theta(\alpha,\beta) = 0$, where estimator labels can be randomly reassigned without changing the distribution of performance differences. 

\textbf{Permutation Test Framework.} Under the null hypothesis $H_0: \theta(\alpha,\beta) = 0$ (equal performance), estimators $\alpha$ and $\beta$ have identical expected performance across studies. This implies that the assignment of estimator labels to the observed performance differences is exchangeable under $H_0$.

Our test procedure works as follows:
\begin{enumerate}
\item \textit{Observed Test Statistic}: Compute the observed difference in performance:
\begin{equation}
\hat{\theta}_{\text{obs}} = \hat{M}_\beta - \hat{M}_\alpha = \frac{1}{N} \sum_{s=1}^N \left(\hat{\epsilon}^{(s)}_\beta - \hat{\epsilon}^{(s)}_\alpha\right)
\end{equation}

\item \textit{Permutation Distribution}: For each of the $2^N$ possible ways to assign estimator labels $\{\alpha, \beta\}$ to the pair of performance estimates $\{(\hat{\epsilon}^{(s)}_\alpha, \hat{\epsilon}^{(s)}_\beta)\}_{s=1}^N$, compute the corresponding test statistic. Let $\pi \in \Pi$ denote a permutation where $\Pi$ is the set of all possible label assignments.

\item \textit{Permuted Test Statistics}: For permutation $\pi$, define:
\begin{equation}
\hat{\theta}_\pi = \frac{1}{N} \sum_{s=1}^N \left(\hat{\epsilon}^{(s)}_{\pi(\beta),s} - \hat{\epsilon}^{(s)}_{\pi(\alpha),s}\right)
\end{equation}
where $\pi(\alpha)$ and $\pi(\beta)$ represent the permuted estimator assignments for study $s$.

\item \textit{P-value Computation}: The two-sided p-value is:
\begin{equation}
p = \frac{1}{|\Pi|} \sum_{\pi \in \Pi} \mathbf{1}[|\hat{\theta}_\pi| \geq |\hat{\theta}_{\text{obs}}|]
\end{equation}
\end{enumerate}

% \textbf{Computational Implementation.} Since $2^N$ can be computationally prohibitive for large $N$, we use Monte Carlo approximation by randomly sampling $B$ permutations (typically $B = 10,000$):

% \begin{equation}
% \hat{p} = \frac{1}{B} \sum_{b=1}^B \mathbf{1}[|\hat{\theta}_{\pi_b}| \geq |\hat{\theta}_{\text{obs}}|]
% \end{equation}
% where $\{\pi_b\}_{b=1}^B$ are randomly sampled permutations.

\subsection{Weighted Metrics for Robustness}

While our cross-fitting approach provides consistent identification of dominant estimators asymptotically, finite-sample performance may be sensitive to outlier studies with unusually large outcome variability. Such high variability can arise from implementation issues, design problems, or genuinely different populations that may not be representative of the broader family of studies under consideration. To address this concern while maintaining statistical rigor, we extend our framework to incorporate study-specific weights that reflect the reliability and representativeness of each study's contribution to our estimator comparison.

We modify our overall performance metric to $M_{\alpha,w} = \mathbb{E}_S[w^{(s)} \cdot \epsilon^{(s)}_\alpha]$, where $w^{(s)} \geq 0$ is a study-specific weight. A natural weighting scheme assigns higher weight to studies with more precise treatment effect estimates through inverse-variance weighting: $w^{(s)} = (\text{Var}(\hat{\tau}^{(s)}_{\text{DM}}))^{-1}$. This approach has several desirable properties: studies with larger samples receive higher weight reflecting their greater statistical power, studies with lower outcome variance receive higher weight as they provide more precise estimates, and studies with severe implementation issues reflected in extremely high variance are automatically down-weighted.

To implement weighted comparisons within our cross-fitting framework, we estimate the weights using the full data in each study. For study $s$, we estimate $\hat{w}^{(s)}$ using the sample variance. The weighted cross-fitted performance estimate becomes $\hat{\epsilon}^{(s)}_{\alpha,w} = \frac{\hat{w}^{(s)}}{K} \sum_{k=1}^K  \Delta(\hat{\tau}^{(s)}_{\alpha,-k}, \hat{\tau}^{(s)}_{\text{DM},k})$, and the weighted test statistic is $\hat{\theta}_w(\alpha,\beta) = \frac{1}{N} \sum_{s=1}^N (\hat{\epsilon}^{(s)}_{\beta,w} - \hat{\epsilon}^{(s)}_{\alpha,w})$.

The weighted approach provides a principled middle ground between purely rank-based comparison (counting only the number of studies where one estimator outperforms another, which ignores magnitude of performance differences) and unweighted comparison (treating all studies equally, which can be dominated by high-variance outliers). The inverse-variance weighting naturally balances these concerns by giving appropriate weight to each study's contribution based on its precision, while still incorporating information from all studies in the family.

% \subsection{Algorithm for K-Estimator Comparison and Selection}

% To systematically identify the optimal estimator across multiple candidates, we implement a tournament-style approach where estimators compete in pairwise comparisons, with winners advancing to subsequent rounds until determining the overall best-performing estimator. Rather than comparing all possible pairs (requiring $k(k-1)/2$ comparisons for $k$ estimators), we implement a more efficient knockout tournament approach requiring only $k-1$ comparisons. This procedure is summarized in Algorithm \ref{alg:estimator_tournament} in the Appendix.

% An alternative approach for multiple pairwise comparisons would be the Copeland/Borda counting method, which tallies the number of wins for each estimator \cite{saari1996copeland}. However, our tournament approach offers two key advantages: computational efficiency (requiring only $k-1$ comparisons rather than $k(k-1)/2$) and a convenient visualization structure for presenting results, as shown in Figure~\ref{fig:tournament}.

% The knockout tournament ensures we identify the best overall performer efficiently while maintaining a clear record of how estimators performed against their competition, providing valuable insights beyond simply identifying the winner.
        \section{Amazon Supply Chain Optimization Technologies Experiments}\label{sec:amazon}

We demonstrate our framework through analysis of randomized controlled trials conducted at Amazon to evaluate supply chain optimization policies. These experiments provide an ideal testbed for comparing treatment effect estimators across a large family of related studies with shared outcome measures and analytical objectives. 

\subsection{Data Description}
Our analysis leverages a sample of 556 RCTs conducted between 2016 and 2025 across diverse geographies, marketplaces, and product categories. Most trials compare a single treatment against control; for multi-arm trials, we construct pairwise comparisons between each treatment and control. While there is a diversity of treatments under study, the common theme for them is that they implement a change to Amazon’s supply chain network. Importantly, these changes are evaluated independently through randomized experiments. The data considered in this study include RCTs with independent units and treatments, satisfying all of the requirements of the methodology presented in this paper. For outcome variable we focus on a common financial outcome. However, the method is readily applicable for any outcome variable for MSE estimation and for any profit-related metric in terms of our regret derivations.
% In the Appendix we include (largely unchanged) results on a larger set of experiments where some units may belong to more than one experiment  --- thus, violating assumption~\ref{ass:non_overlap} --- showcasing the robustness of the methodology.  
 
%\paragraph*{Outcome Variable.} We focus on Long Term Free Cash Flow (LTFCF), a composite metric capturing: (i) profit, (ii) customer lifetime value from availability and delivery speed, and (iii) opportunity cost of holding inventory. 

\paragraph*{Estimators.} We evaluate four base estimators: (1) Difference of Means (DM), (2) Ordinary Least Squares with covariate adjustment (OLS), (3) Linear T-learner with separate models per treatment arm, and (4) Weighted Least Squares (WLS) with weights $w_i = 1/(Y_{\text{pre},i} - \bar{Y}_{\text{pre}})^{2}$ that emphasize units near the pre-treatment mean. Given the heavy-tailed nature of this financial metric, we combine each base estimator with winsorization. Particularly, we consider following seven levels of winsorization: $p \in \{0, 0.001, 0.005, 0.01, 0.025, 0.05, 0.1\}$. We denote combinations as \textit{Estimator\_Level} (e.g., "DM\_0.01" for difference-of-means with 1\% winsorization). This yields 28 total estimator variants.

\paragraph*{Evaluation Metrics and Objectives.} Amazon's experimentation serves two objectives: (i) launch decisions and (ii) uplift measurement. For launch decisions, we use the regret metric $\Delta_{\text{Regret}}$ to identify estimators minimizing opportunity costs from false positives (launching ineffective policies) and false negatives (missing beneficial policies). For uplift measurement, we use Mean Squared Error ($\Delta_{\text{MSE}}$) to prioritize accurate and precise effect estimation. This dual evaluation demonstrates how optimal estimation strategies depend on analytical objectives.

\subsection{Results and Analysis}

Our analysis reveals substantial differences in estimator performance across evaluation objectives, confirming that optimal strategies are context-dependent. Figures~\ref{fig:amazon_mse} and \ref{fig:amazon_regret} present results for MSE and regret-based evaluation.

\paragraph*{MSE-Based Performance for Inference.} For mean squared error (Figure~\ref{fig:amazon_mse}), Linear T-learner with 0.5\% winsorization (Linear-T\_0.005) achieves minimum MSE. However, optimality varies by winsorization level: at levels below 0.5\%, OLS becomes optimal. MSE displays the expected smooth, convex shape with a global minimum reflecting the bias-variance tradeoff. As winsorization increases, performance differences between DM, OLS, and Linear T-learner diminish, suggesting that winsorization's variance reduction eliminates the advantage of covariate adjustment -- though at the cost of increased bias.

\paragraph*{Regret-Based Performance for Decision-Making.} Regret-based evaluation\footnote{We normalize regret so that the value at critical threshold 1.96 for DM\_0 equals 1.} (Figure~\ref{fig:amazon_regret}) reveals a strikingly different ranking. Difference-of-means without winsorization (DM\_0) achieves best performance, while Linear-T\_0.005, the estimator with minimum MSE, performs suboptimally here.

This reversal illuminates a fundamental insight: \textit{estimators minimizing estimator error need not minimize decision costs.} The DM estimator's unbiased, symmetric sampling distribution better aligns with minimizing decision errors despite higher variance. Linear-T\_0 ranks second for regret, with OLS a close third, suggesting these covariate-adjusted methods offer reasonable compromises when objectives are mixed or uncertain.

\paragraph*{Practical Implications.} These results highlight that principled estimator selection requires explicit consideration of analytical objectives. For Amazon's context, we recommend DM\_0 for launch decisions and Linear-T\_0.005 for uplift measurement, with T-Learner\_0 as a reasonable default when objectives are uncertain. The performance reversals between metrics underscore why ad hoc estimator selection may yield suboptimal outcomes in practice.

\begin{figure}
    \centering
    \includegraphics[width=\linewidth]{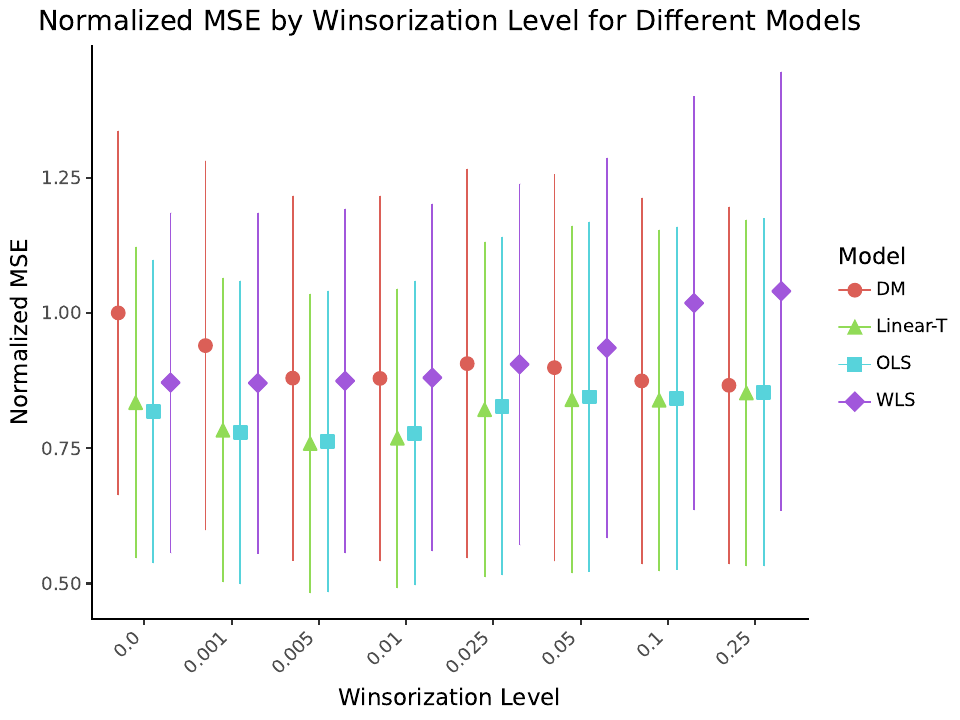}
    \caption{Mean Squared Error performance across treatment effect estimators in Amazon supply chain experiments as winsorization level vary. Values normalized to DM with no winsorization.}
    \label{fig:amazon_mse}
\end{figure}

\begin{figure}
    \centering
    \includegraphics[width=0.9\linewidth]{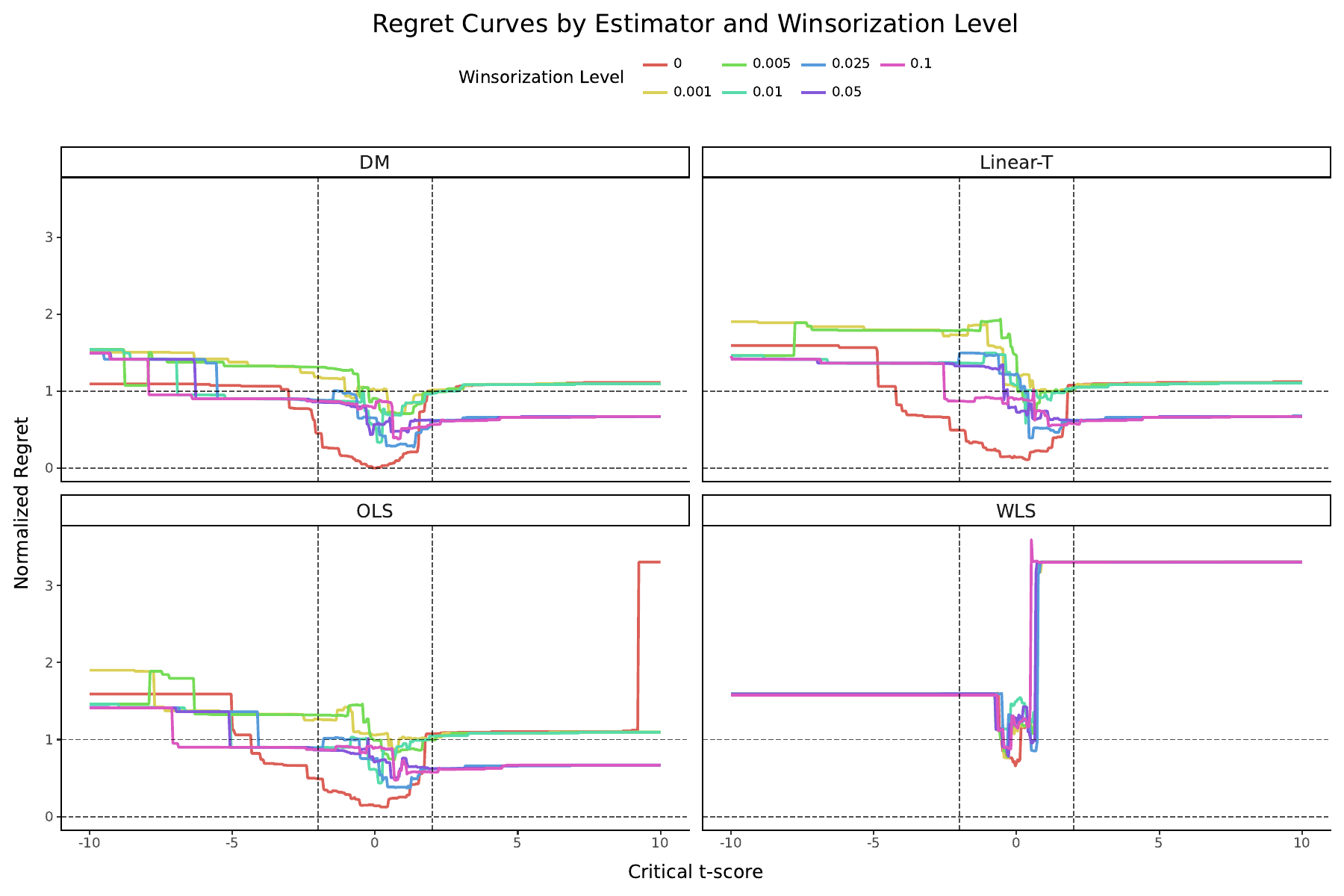}\\
    (a)\\
    \includegraphics[width=0.8\linewidth]{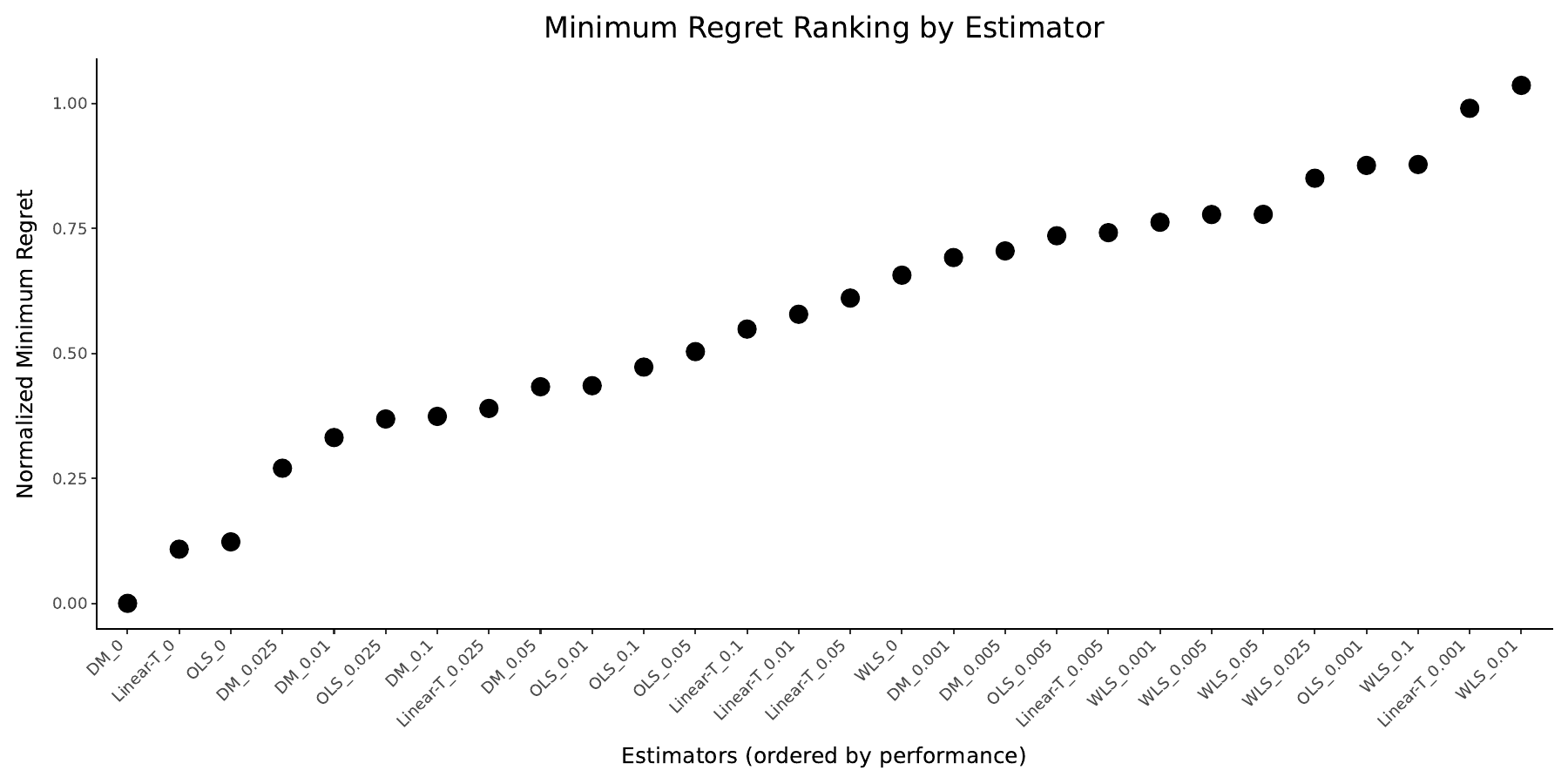}\\(b)
    \caption{(a) Minimum ``regret'' performance across treatment effect estimators in Amazon supply chain experiments. (b) Normalized ``regret'' for different values of critical t-score $c$ across estimators.}
    \label{fig:amazon_regret}
\end{figure}
        \section{Strengthening Democracy Challenge}
\label{sec:case_study}
Next, we apply our framework to the publicly available Strengthening Democracy Challenge (SDC) data. Here, the primary goal is inferential: understanding intervention effects on democratic attitudes. We use mean squared error (MSE) as our evaluation metric ($\Delta$).

\subsection{Data Description}

The SDC provides an ideal testbed featuring 25 interventions targeting similar outcomes with consistent baseline covariates \citep{voelkel2024mega}. The experiment recruited 32,059 U.S. adults in January 2022 via an online platform, randomly assigning approximately 1,200 participants to each of 25 treatment conditions or control. Interventions ranged from perspective-taking exercises to media literacy training, each designed to be brief (2-8 minutes) and scalable. The control condition exposed participants to an unrelated article.

\paragraph*{Outcome Measures.} We focus on two constructs: Support for Partisan Violence (SPV) measured through three 7-point scale items ($SPV_1$: threatening messages, $SPV_2$: online harassment), and Partisan Animosity (PA) measured through feeling thermometers toward Republicans ($PA\text{-}Fth_{Rep}$) plus a Dictator Game allocation ($PA_{DG}$). Each intervention-control comparison constitutes a separate RCT, yielding 25 RCTs per outcome for systematic estimator evaluation.

\subsection{Results and Analysis}

Our analysis reveals substantial variation in optimal estimators across outcomes and evaluation metrics, reinforcing that estimator selection must align with both data characteristics and analytical objectives.

\paragraph*{MSE-Based Performance for Inference.} Figure~\ref{fig:mse_sdc} displays mean squared error across estimators and winsorization levels. For $SPV_1$ and $SPV_2$, all estimators achieve best performance at high winsorization levels (around 10\%), with OLS\_0.1 emerging as optimal. Conditional on winsorization level, estimators perform nearly identically, suggesting weak associations between outcomes and pre-treatment features, low treatment effect heterogeneity, and heavy-tailed distributions characteristic of social science outcomes.

In contrast, $PA\text{-}Fth_{Rep}$ shows difference-in-means performing substantially worse than covariate-adjusted estimators (IPW, OLS, linear T-learner), with winsorization providing minimal benefit even at 10\%. This pattern indicates stronger covariate-outcome relationships where adjustment methods yield precision gains. For $PA_{DG}$, difference-in-means performs comparably to other estimators, but substantial winsorization (5-10\%) significantly improves performance across all methods, revealing outcome-specific sensitivity to extreme values.

These patterns demonstrate that optimal estimators and hyperparameters—including winsorization levels—depend critically on outcome characteristics. Heavy-tailed outcomes with weak covariate relationships benefit from aggressive winsorization, while outcomes with strong covariate associations gain more from adjustment methods than from trimming extremes.

\paragraph*{Regret-Based Performance for Decision-Making.} When evaluated using regret (Figures~\ref{fig:regret_sdc} and \ref{fig:min_regret_sdc}), estimator differences largely disappear. Across all outcomes—$SPV_1$, $SPV_2$, $PA\text{-}Fth_{Rep}$, and $PA_{DG}$—estimators perform nearly identically, with DM\_0 achieving minimum normalized regret.

Figure~\ref{fig:regret_sdc} reveals two distinct regret curve patterns. For $SPV_1$ and $SPV_2$, symmetric regret curves with minima near zero t-critical indicate that optimal rollout decisions require only positive effect estimates. In contrast, $PA\text{-}Fth_{Rep}$ and $PA_{DG}$ exhibit sigmoid-shaped curves: low t-critical thresholds impose minimal costs (suggesting many positive treatment effects), while high t-critical values incur substantial regret from failing to implement beneficial treatments.

\paragraph*{Implications.} The dramatic divergence between MSE and regret rankings reinforces our framework's core insight: optimal estimation strategies depend on analytical objectives. For inferential goals with these outcomes, covariate adjustment with outcome-specific winsorization minimizes MSE. For rollout decisions, simple difference-in-means without winsorization minimizes opportunity costs. Organizations conducting experiments for both purposes must explicitly prioritize objectives or accept compromise solutions that perform reasonably—though not optimally—across metrics.

\begin{figure}[h]
    \centering
    \begin{tabular}{cc}
         \includegraphics[width=0.45\linewidth]{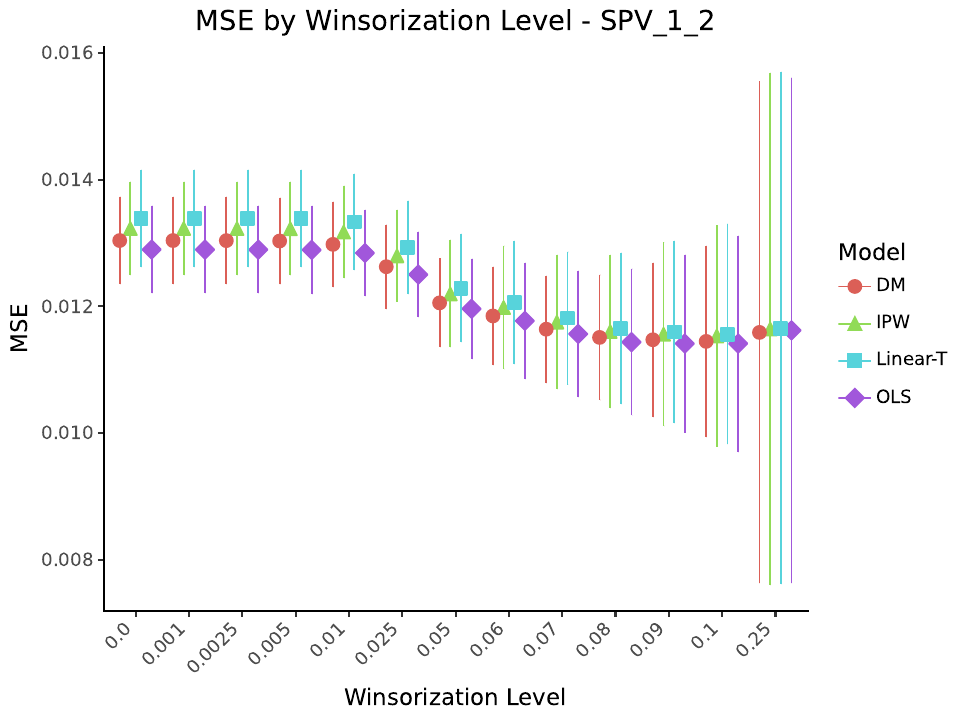} & \includegraphics[width=0.45\linewidth]{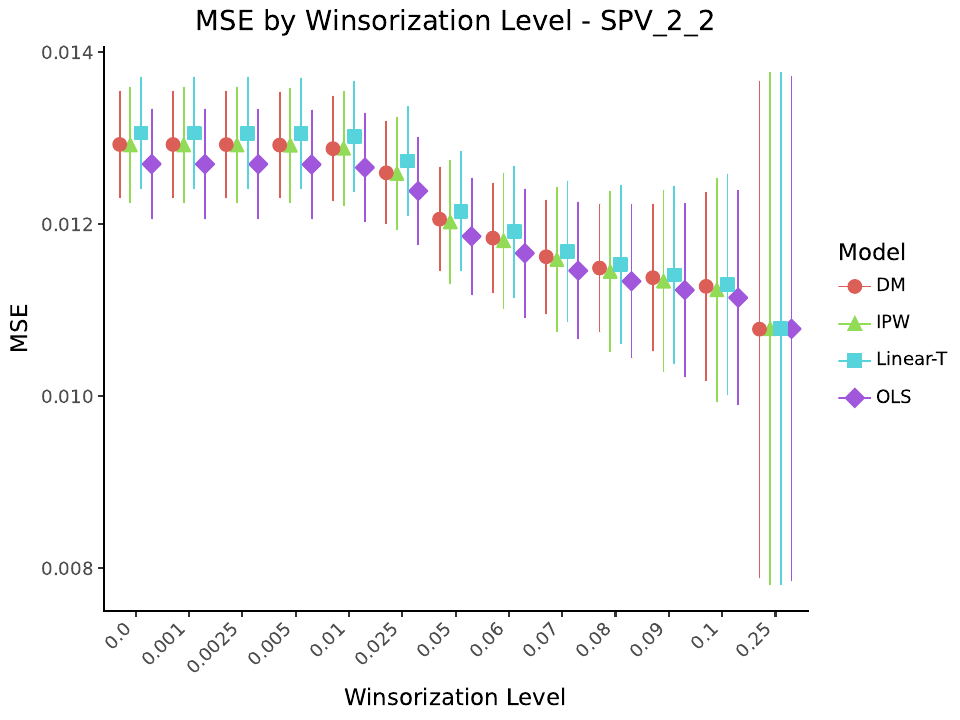} 
         \\
         (a) $SPV_1$ & (b) $SPV_2$ 
         \\
         \includegraphics[width=0.45\linewidth]{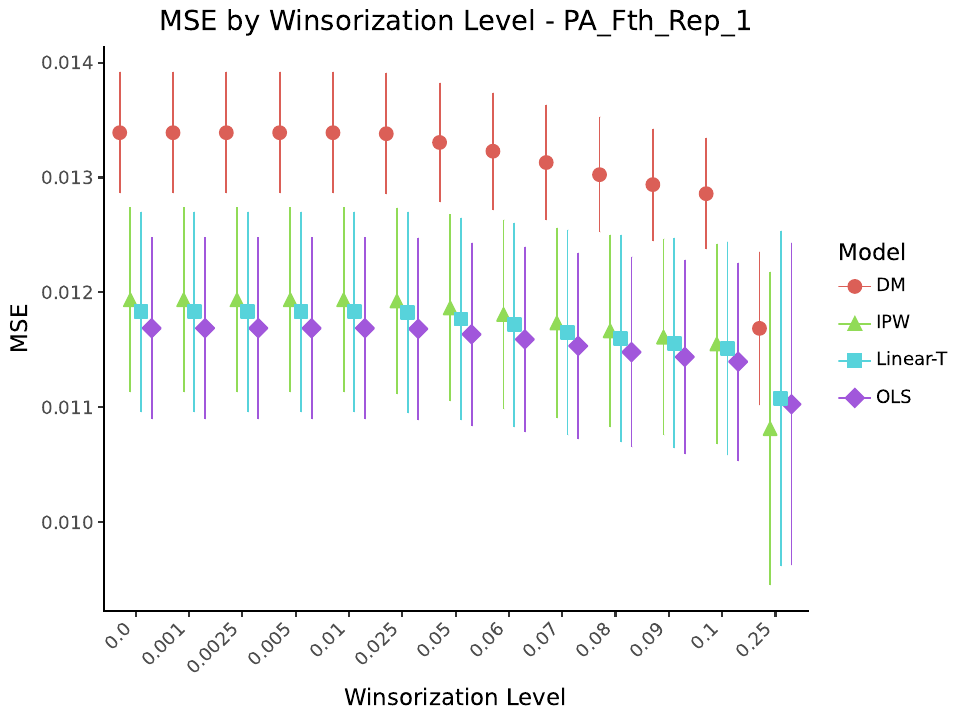} &
         \includegraphics[width=0.45\linewidth]{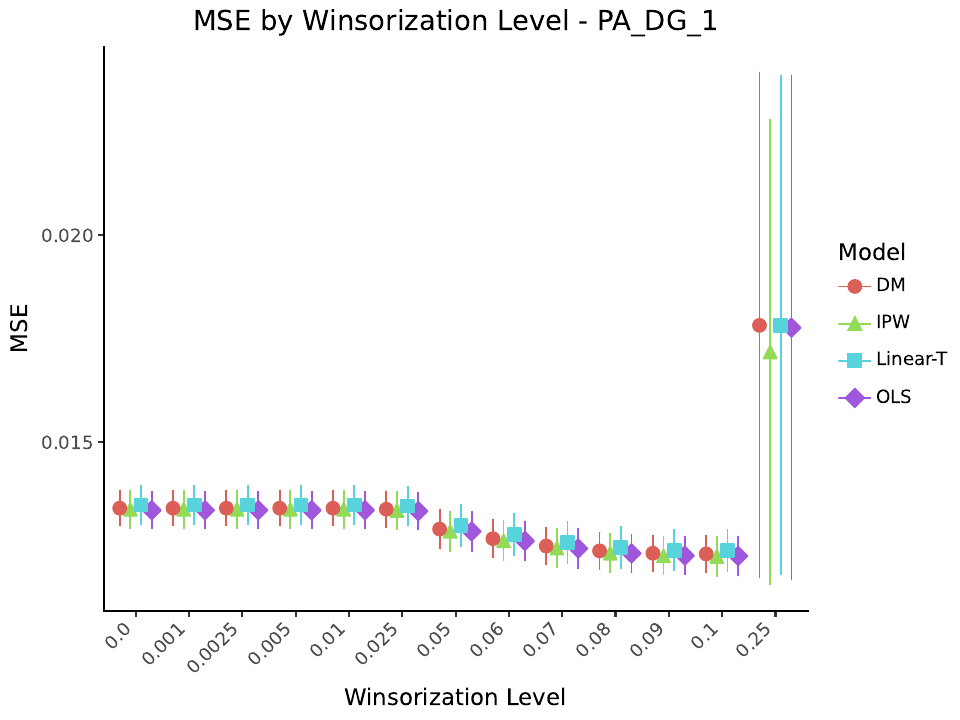}\\
         (c) $PA\text{-}Fth_{Rep}$ & (d) $PA_{DG}$
    \end{tabular}
    \caption{Mean squared error and standard deviation for six estimators across outcomes. Lower values indicate better performance.}
    \label{fig:mse_sdc}
\end{figure}

\begin{figure}[h]
    \centering
    \begin{tabular}{cc}
         \includegraphics[width=0.45\linewidth]{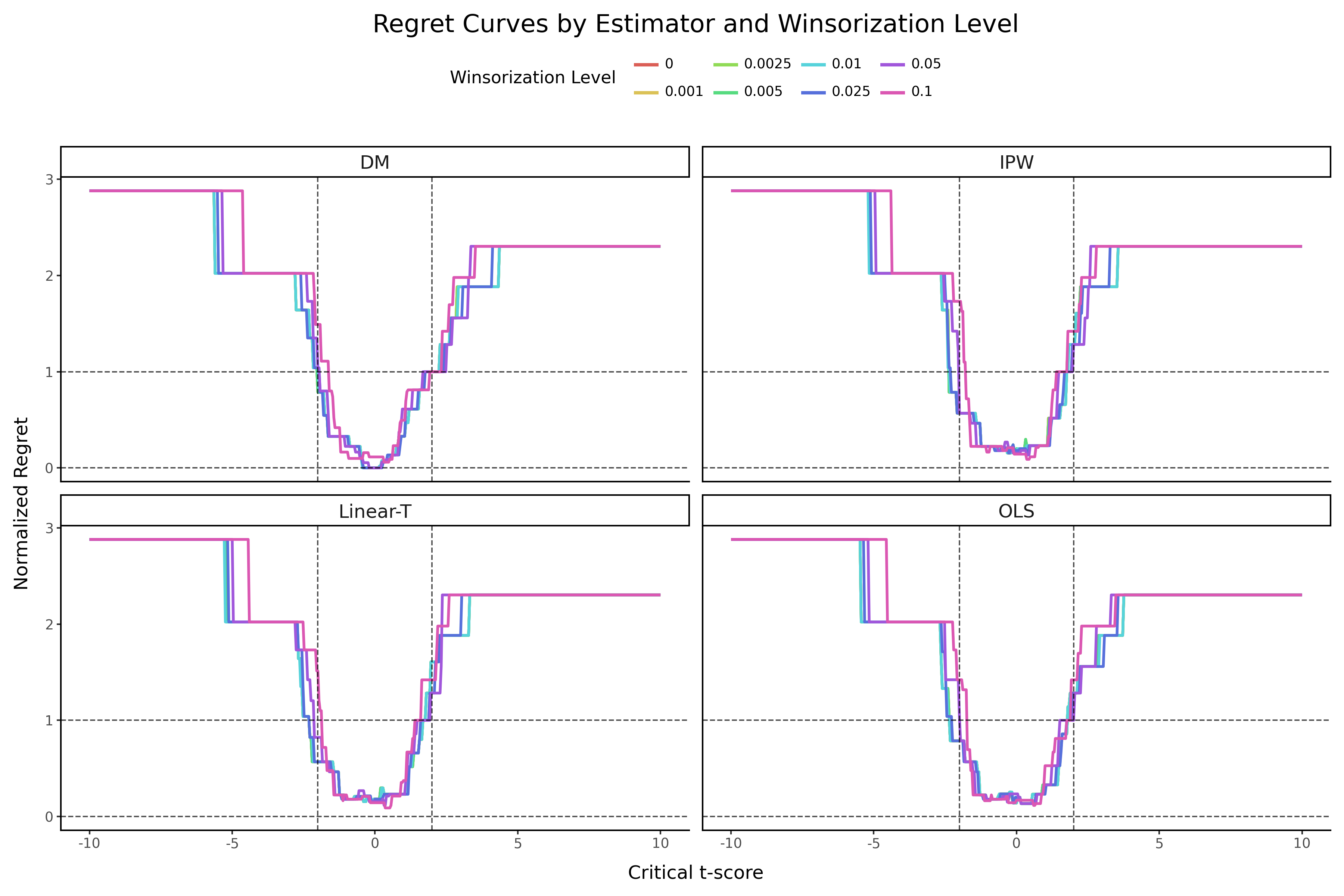} & \includegraphics[width=0.45\linewidth]{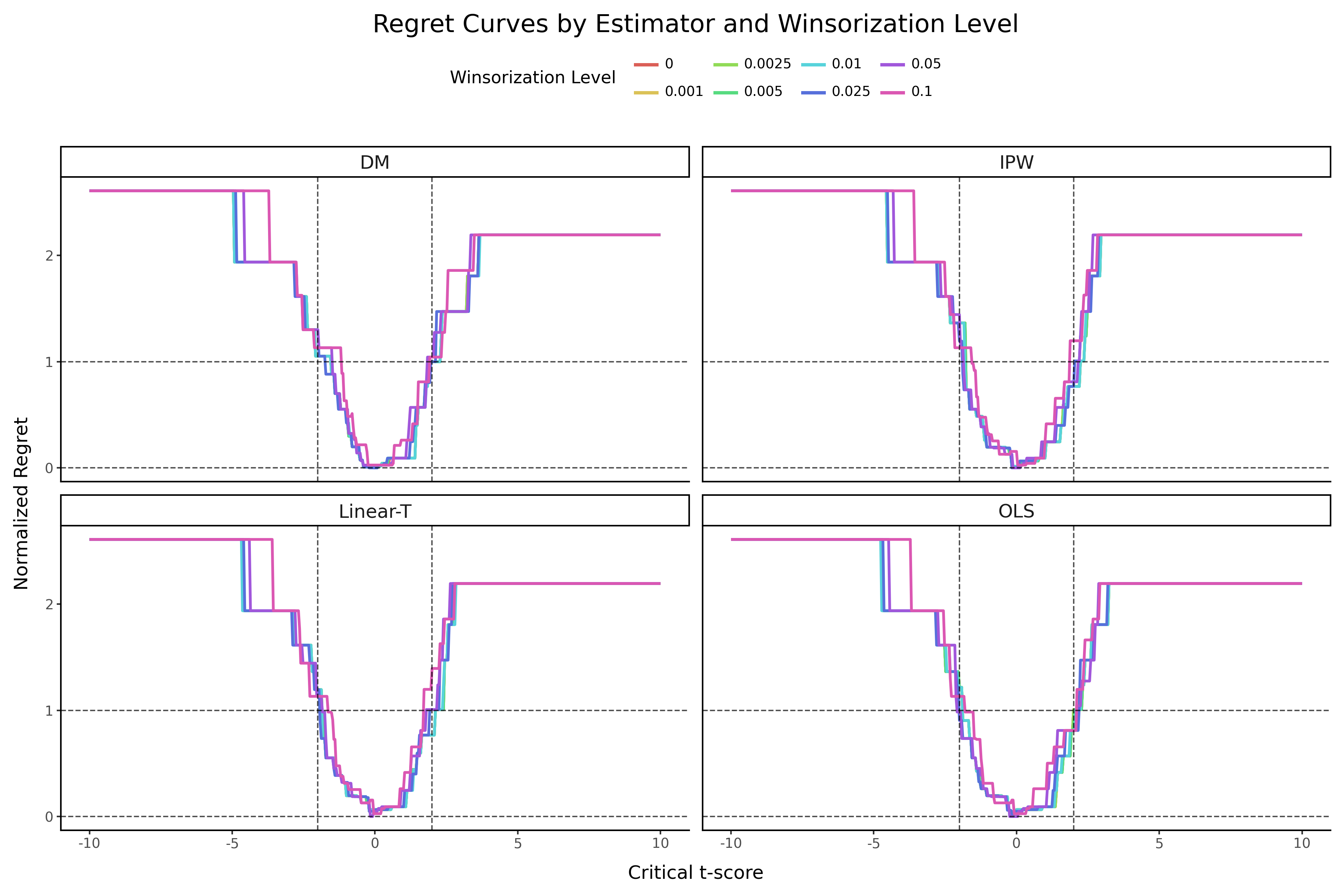}
         \\
         (a) $SPV_1$ & (b) $SPV_2$
         \\
         \includegraphics[width=0.45\linewidth]{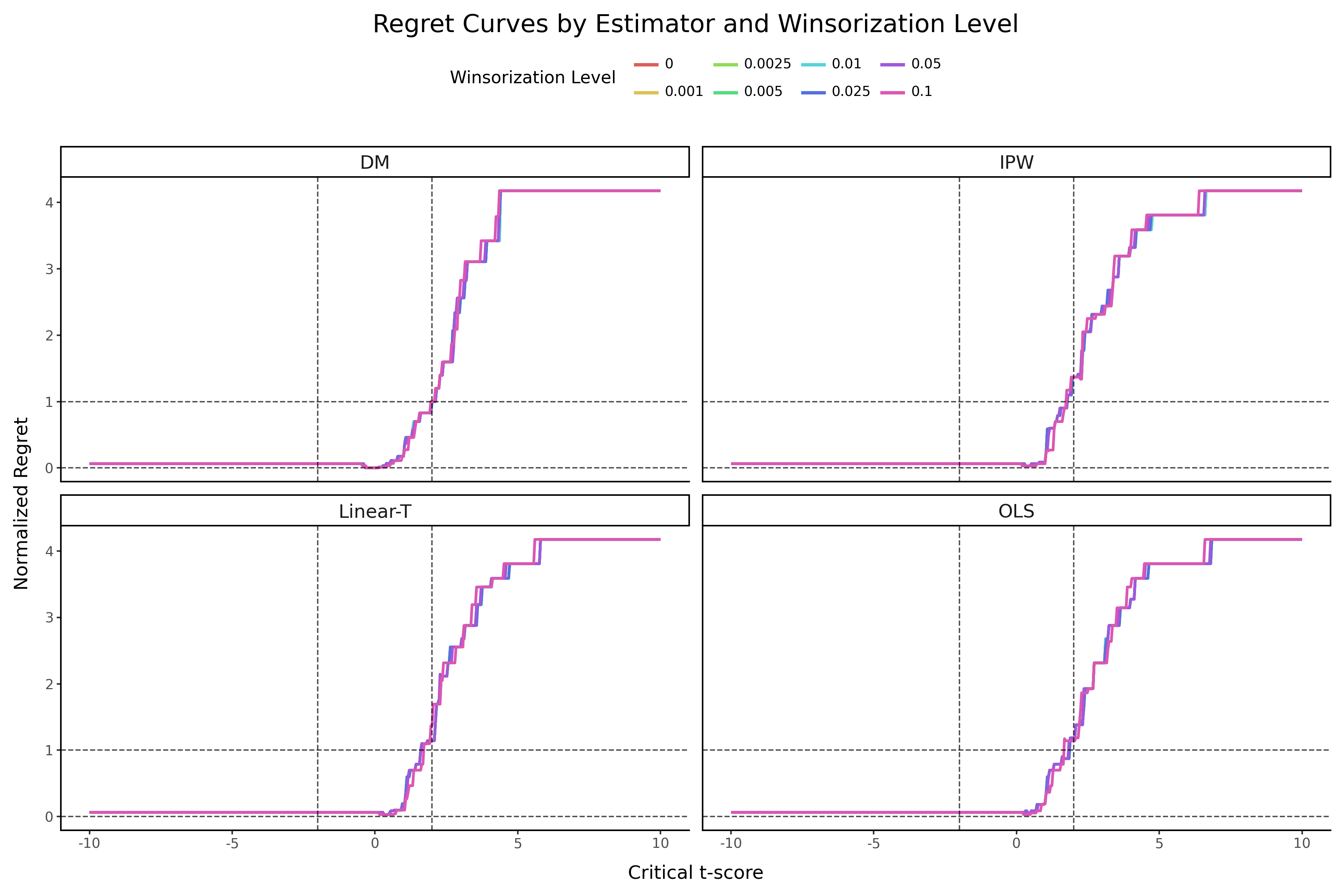} &
         \includegraphics[width=0.45\linewidth]{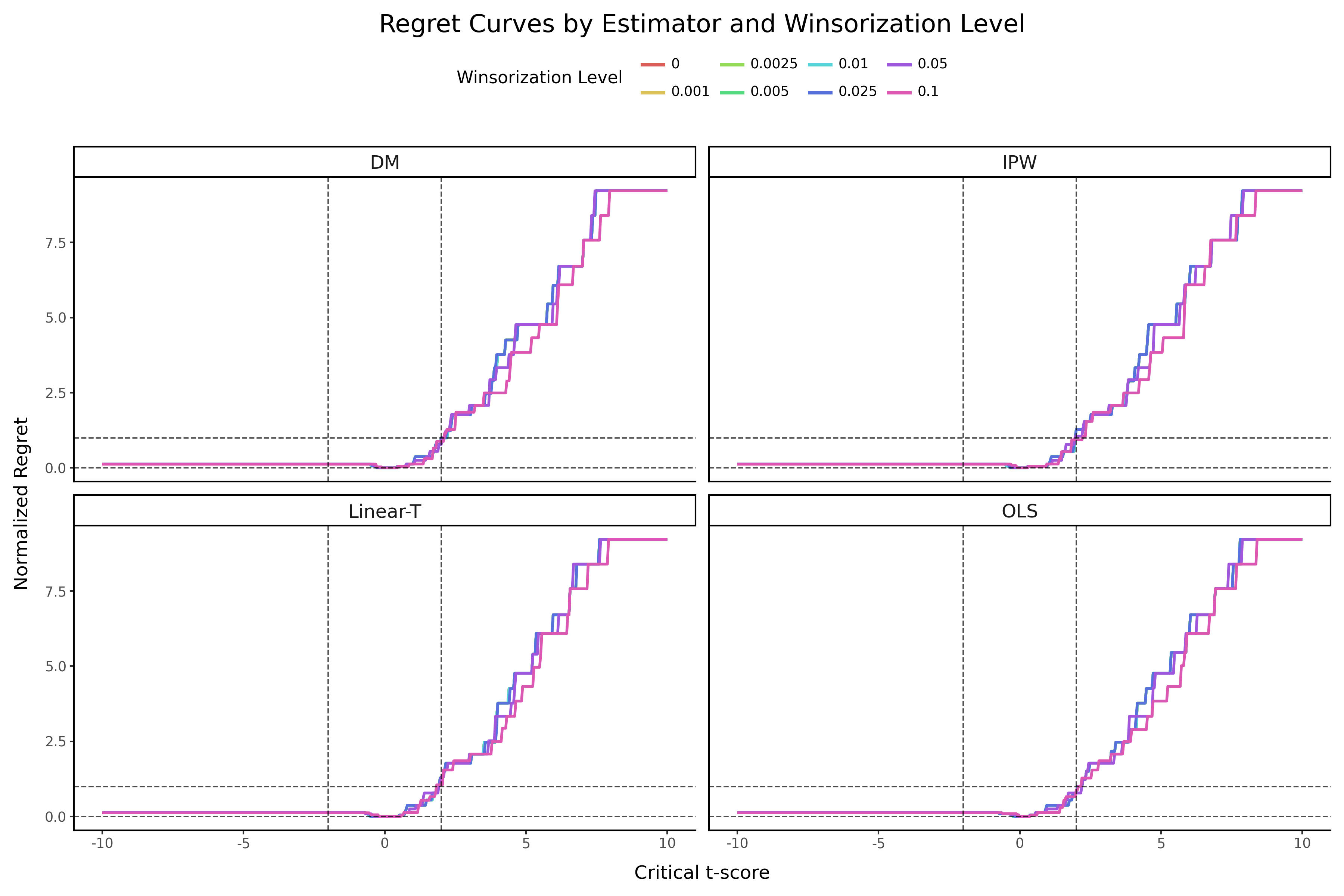}\\
         (c) $PA\text{-}Fth_{Rep}$ & (d) $PA_{DG}$
    \end{tabular}
    \caption{Regret Curves estimators as a function of critical-t across outcomes. Lower values indicate better performance.}
    \label{fig:regret_sdc}
\end{figure}

\begin{figure}[h]
    \centering
    \begin{tabular}{cc}
         \includegraphics[width=0.45\linewidth]{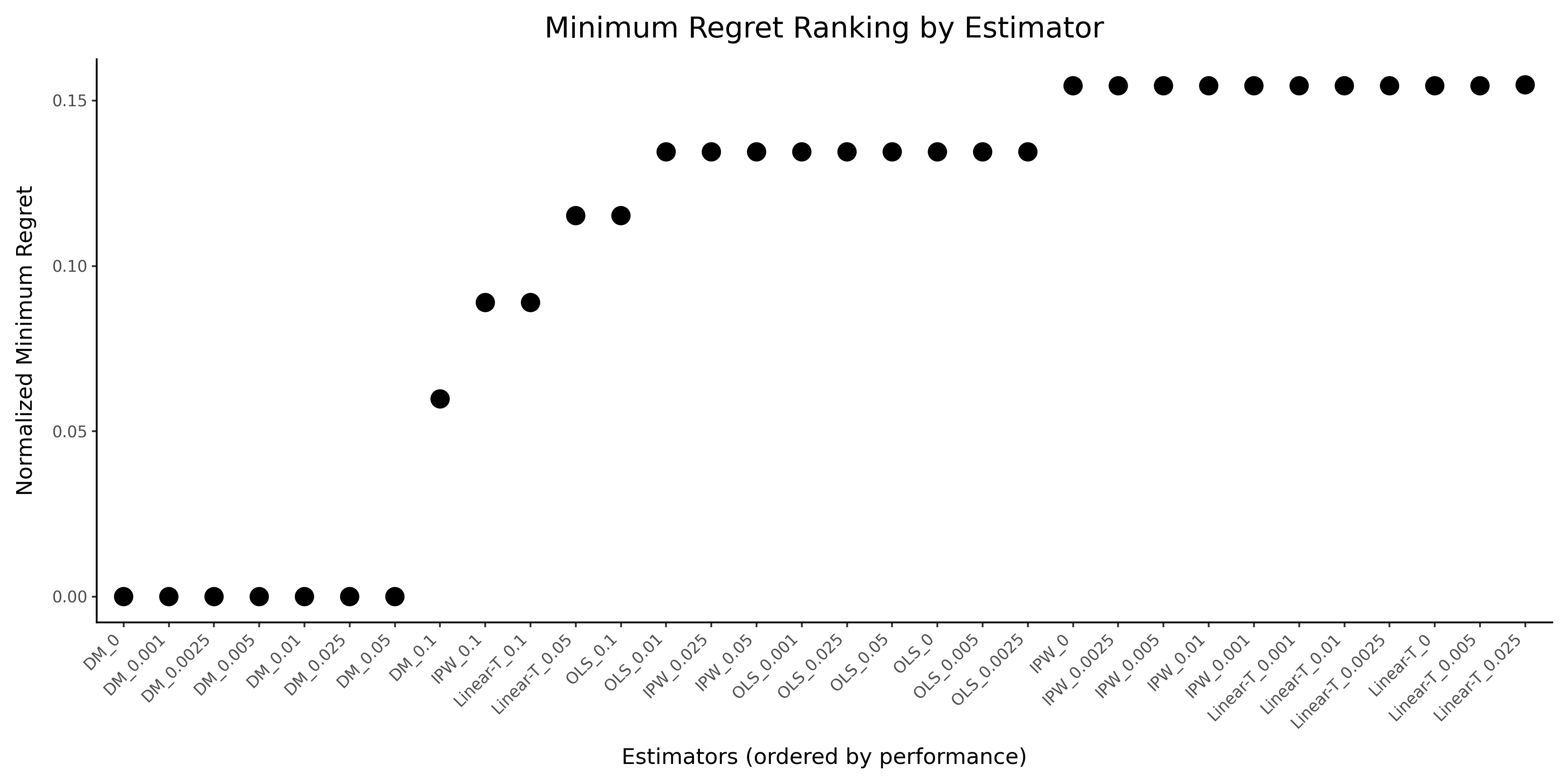} & \includegraphics[width=0.45\linewidth]{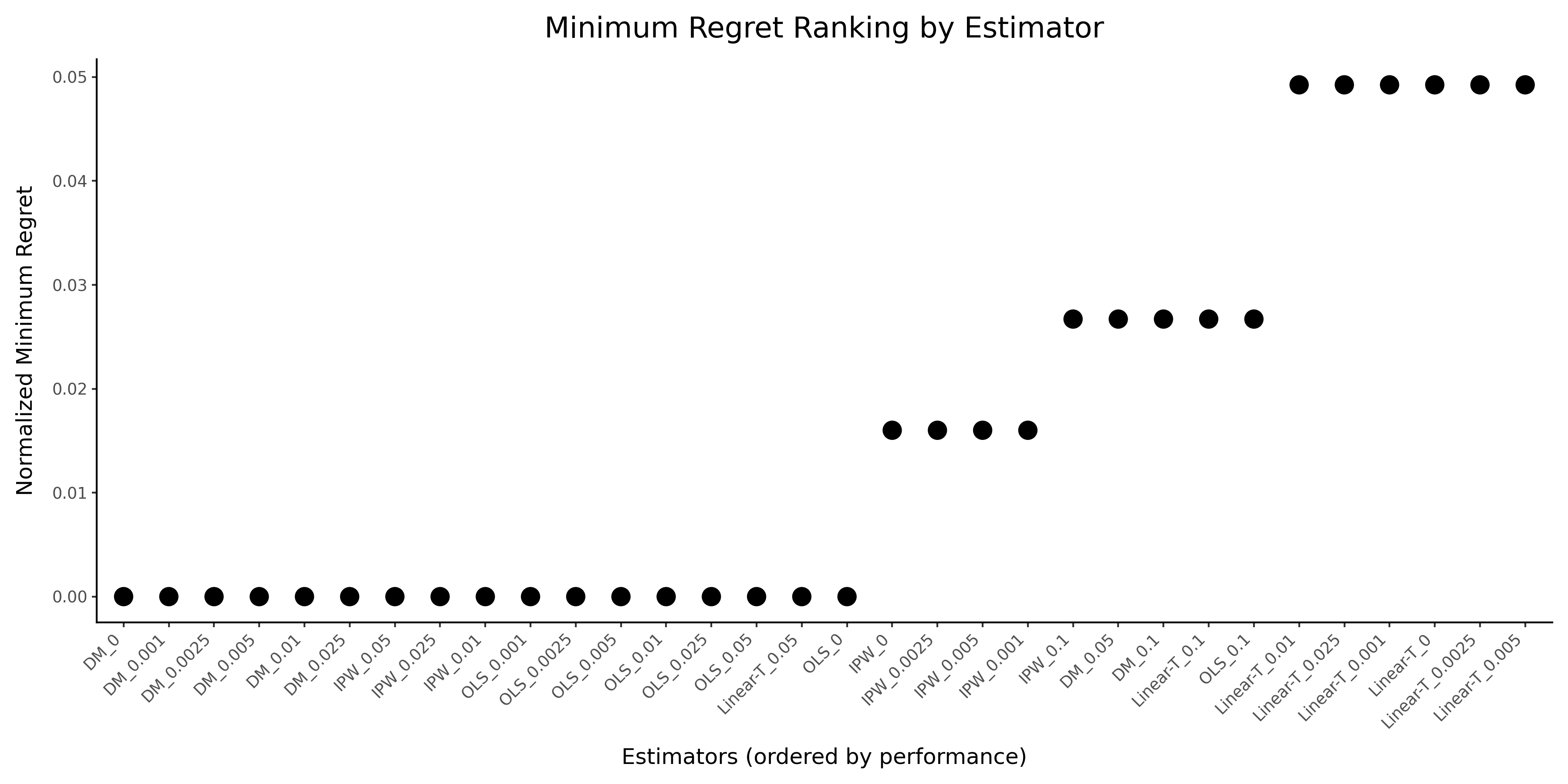}
         \\
         (a) $SPV_1$ & (b) $SPV_2$
         \\
         \includegraphics[width=0.45\linewidth]{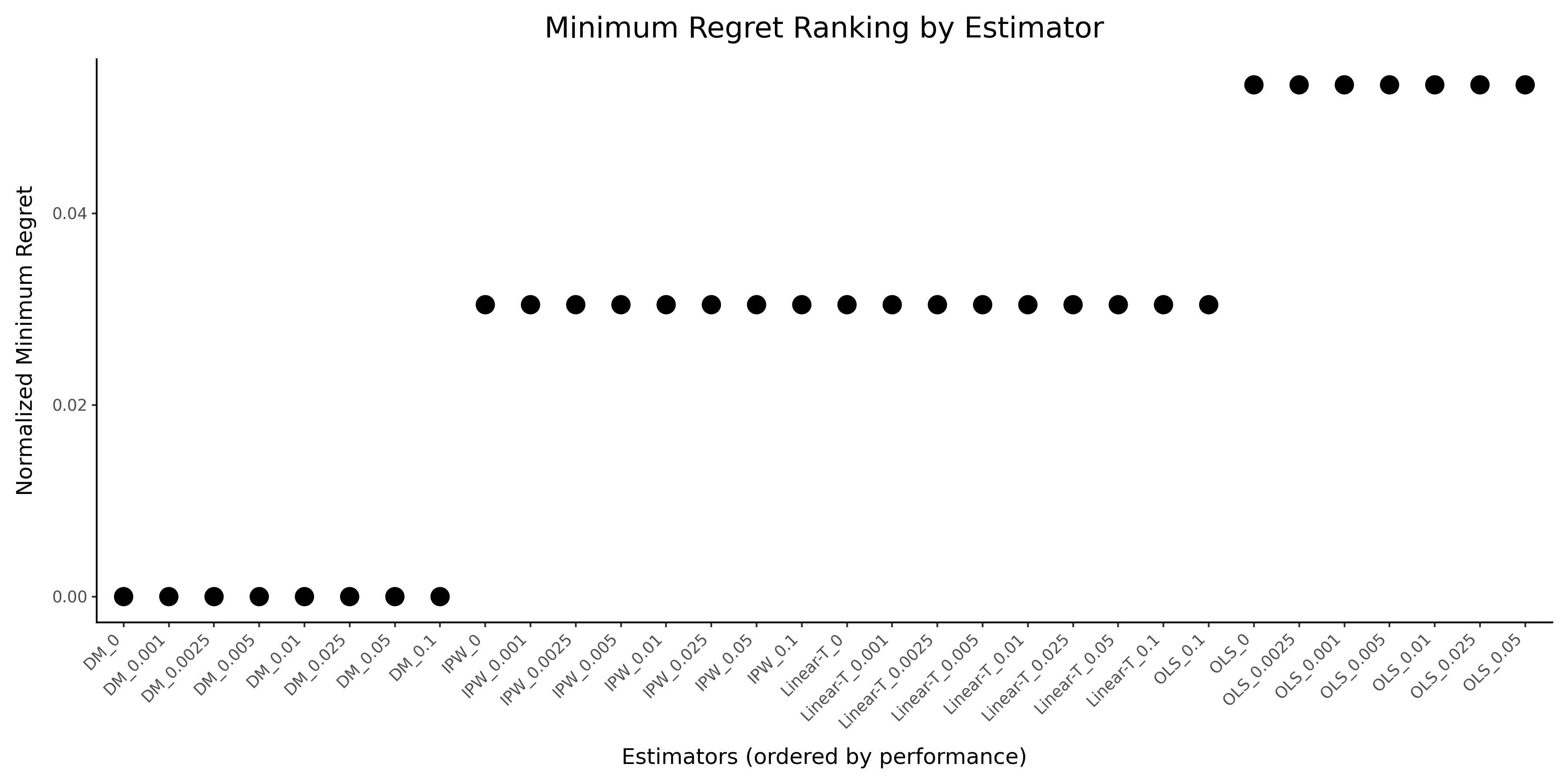} &
         \includegraphics[width=0.45\linewidth]{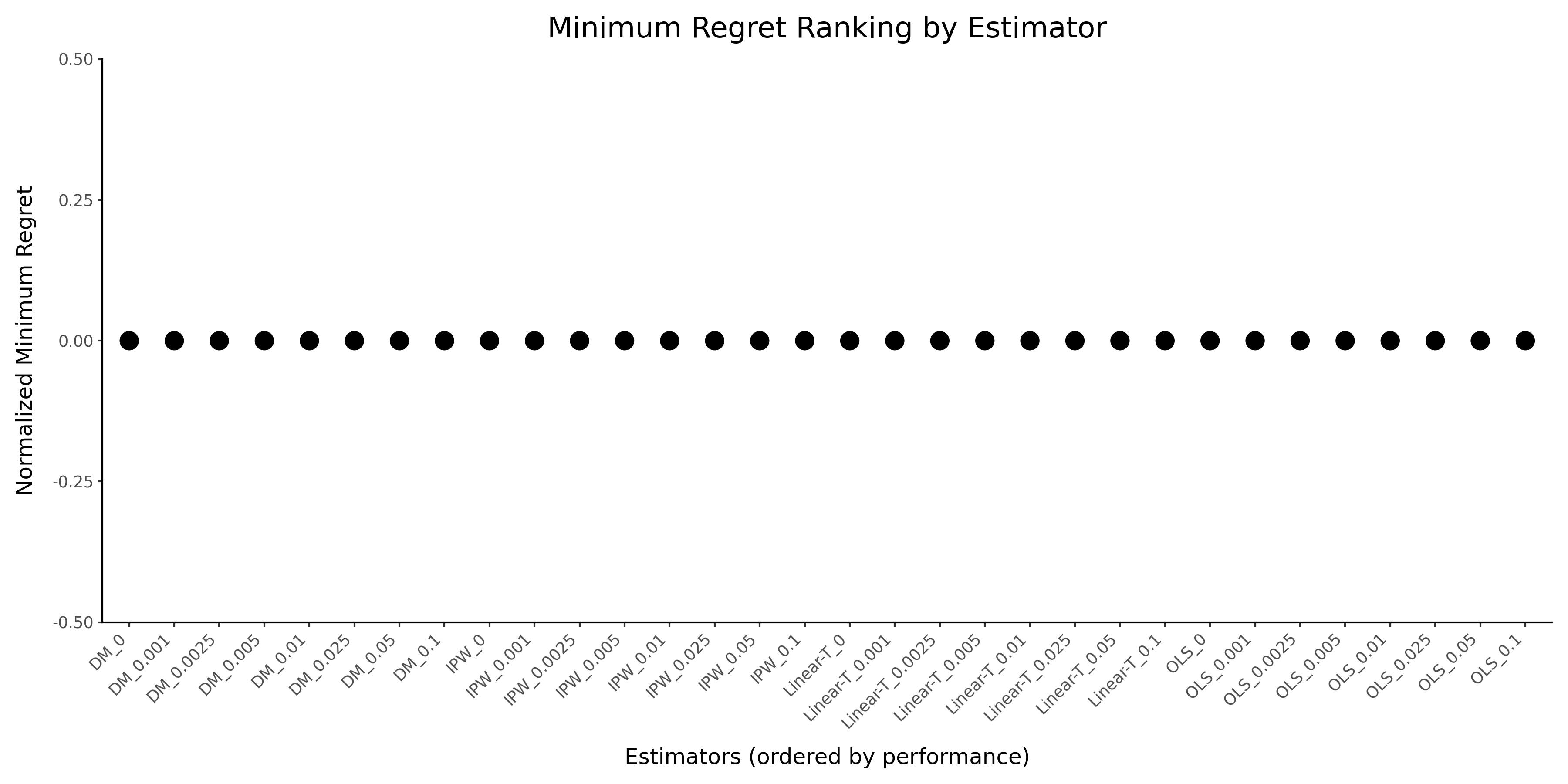}\\
         (c) $PA\text{-}Fth_{Rep}$ & (d) $PA_{DG}$
    \end{tabular}
    \caption{Minimum regret achieved by each estimator across outcomes.}
    \label{fig:min_regret_sdc}
\end{figure}

	%!TEX root = main.tex
\section{Conclusion and Discussion}
\label{sec:conclusion}

This paper addresses a key challenge in experimental design and causal inference: how to systematically select treatment effect estimators when faced with multiple reasonable methodological choices. 

Our work makes several key contributions to the methodological literature on causal inference and experimental design. We develop a principled framework for estimator comparison that avoids the pitfalls of both ad hoc selection. We also demonstrate that optimal estimator selection depends critically on the intended use of treatment effect estimates. Through our Amazon supply chain and SDC case studies, we show that estimators optimized for inferential precision (measured by MSE) can perform poorly for decision-making applications (measured by regret), and vice versa. 

Our findings juxtapose several prevailing assumptions in the experimental literature. The performance reversals we observe between MSE and regret evaluation suggest that the community's focus on p-values may not always translate to optimal decision-making performance. This disconnect is particularly relevant for applications where experiments are primarily conducted to inform operational decisions rather than to estimate parameters for scientific understanding. Our framework also highlights the importance of domain-specific methodological validation. We suggest that organizations and research communities would benefit from developing domain-specific best practices rather than relying solely on general-purpose methodological recommendations.

More broadly, our work contributes to the growing recognition that causal inference methods must be tailored to their intended applications. Just as the choice between randomized trials and observational studies depends on feasibility and ethical considerations, the choice of estimation methods within randomized trials should depend on analytical objectives and domain characteristics.

Several limitations of our current approach suggest important directions for future research. First, our framework currently focuses on treatment effect estimation for categorical treatment in cross-sectional experiment. Extensions to time varying treatment experiments, factorial designs, and more complex treatment structures would broaden the applicability of our approach. Similarly, incorporating methods for handling interference, non-compliance, and other complications common in field experiments would enhance practical relevance.

Second, while our permutation-based inference approach provides robust statistical validity, it may be computationally intensive for very large families of experiments. Developing more efficient computational approaches or approximation methods could facilitate adoption in settings with thousands of experiments.

Third, our current framework treats the choice of evaluation metric (MSE versus regret) as given, but in practice, determining the appropriate metric for a given application may itself be challenging. Future work could develop principled approaches for metric selection or multi-objective optimization approaches that balance multiple evaluation criteria.

Based on our findings, we offer several practical recommendations for researchers and practitioners conducting experimental programs:

Organizations conducting experiments primarily for decision-making should prioritize regret-based evaluation over MSE-based evaluation when selecting estimation methods. Simple methods like difference-of-means may outperform sophisticated alternatives when the goal is minimizing regreet rather than achieving precise parameter estimates.

Researchers should explicitly articulate whether their primary objective is scientific inference or practical application, as this choice should influence methodological decisions. When both objectives are important, methods like OLS that perform reasonably well across evaluation criteria may provide a sensible compromise.

The substantial performance differences we observe across evaluation criteria suggest that new methods should be evaluated using metrics aligned with their intended applications. Method developers should consider multiple evaluation criteria and acknowledge trade-offs between different objectives.

Organizations operating large-scale experimentation platforms should consider implementing our framework to develop domain-specific best practices rather than relying on one-size-fits-all recommendations. Regular validation exercises using accumulated experimental data can inform evolving methodological guidelines.

As experimental methods continue to evolve and new techniques emerge, frameworks like ours will become increasingly important for ensuring that methodological sophistication translates into improved decision-making and scientific understanding. The ultimate goal is not to find the single best method, but to develop principled approaches for selecting methods that are well-suited to specific research contexts and analytical objectives.

        \bibliographystyle{apalike}
        \bibliography{biblio.bib}
        \newpage
        \appendix
        \begin{appendix}
            \section{Treatment Effect Estimators}
\label{sec:estimators}

In this section, overview some commonly used treatment effect estimators. For clarity, we present estimators for binary treatment comparison ($T_a = 1$ versus $T_b = 0$) within a single RCT, with straightforward extension to arbitrary treatment comparisons across studies.

\subsection{Core Estimators and Covariate Adjustment}

\textbf{Difference-of-Means (DM).} The most fundamental approach directly compares sample averages between treatment and control groups:
\begin{equation}
\hat{\tau}_{DM} = \frac{1}{m_{s,1}}\sum_{i:T_i=1}Y_i - \frac{1}{m_{s,0}}\sum_{i:T_i=0}Y_i,
\end{equation}
where $m_{s,1}$ and $m_{s,0}$ are the treatment and control group sizes. While unbiased under randomization, DM can exhibit high variance and fails to leverage pretreatment covariate information.

\textbf{S-Learner and T-Learner.} Machine learning approaches model pretreatment covariates using flexible functional forms. S-learners fit a single model including treatment as a feature: $\hat{Y} = \hat{f}(X, T)$, with treatment effect estimated as $\hat{\tau}_{S} = \hat{f}(X, 1) - \hat{f}(X, 0)$. T-learners fit separate models for each treatment group: $\hat{\mu}_1(X) = \hat{f}_1(X)$ and $\hat{\mu}_0(X) = \hat{f}_0(X)$, with treatment effect $\hat{\tau}_{T} = \hat{\mu}_1(X) - \hat{\mu}_0(X)$. We implement both linear regression and random forest versions. T-learners allow different functional forms between treatment groups, better capturing heterogeneity.

\textbf{Inverse Probability Weighting (IPW) and Augmented IPW.} IPW incorporates covariate information through estimated propensity scores:
\begin{equation}
\hat{\tau}_{IPW} = \frac{1}{m_s}\sum_{i=1}^{m_s} \left(\frac{T_iY_i}{\hat{e}(X_i)} - \frac{(1-T_i)Y_i}{1-\hat{e}(X_i)}\right),
\end{equation}
where $\hat{e}(X_i)$ is the estimated propensity score. Augmented IPW (AIPW) combines IPW with outcome regression for double robustness \citep{chernozhukov2018double}:
\begin{equation}
\hat{\tau}_{AIPW} = \frac{1}{m_s}\sum_{i=1}^{m_s} \left[\frac{T_i(Y_i-\hat{\mu}_1(X_i))}{\hat{e}(X_i)} + \hat{\mu}_1(X_i) - \frac{(1-T_i)(Y_i-\hat{\mu}_0(X_i))}{1-\hat{e}(X_i)} - \hat{\mu}_0(X_i)\right],
\end{equation}
where $\hat{\mu}_t(X_i)$ are estimated outcome models for each treatment group.

\textbf{Causal Forest.} Causal forests extend random forests to target treatment effect heterogeneity by maximizing treatment effect variance across terminal nodes rather than prediction accuracy \citep{wager2018estimation}. They provide a data-driven approach for discovering heterogeneous effects but require larger sample sizes.

\subsection{Variance Reduction Techniques}

The estimators described above are often combined with additional techniques designed to reduce variance and improve precision. These techniques can be applied to any of the core estimators as complementary approaches.

\textbf{Winsorization.} Winsorization addresses heavy-tailed outcome distributions by capping extreme values at specified percentiles, typically the 1st and 99th percentiles. For outcome $Y_i$, winsorized values are:
\begin{equation}
\tilde{Y}_i = \begin{cases}
Q_{0.01} & \text{if } Y_i < Q_{0.01} \\
Y_i & \text{if } Q_{0.01} \leq Y_i \leq Q_{0.99} \\
Q_{0.99} & \text{if } Y_i > Q_{0.99}
\end{cases}
\end{equation}
where $Q_p$ denotes the $p$-th quantile. Any estimator can then be applied to the winsorized outcomes $\{\tilde{Y}_i\}$. While winsorization reduces variance by limiting the influence of extreme values, it introduces bias by changing the estimand from the true population effect to the effect on the winsorized distribution.

\textbf{Weighted Average Treatment Effects.} Rather than estimating the standard average treatment effect, estimators can be repurposed to target weighted average treatment effects where the weight for each unit $i$ is chosen as a function of covariates $w_i = w(X_i)$. The weighted estimand becomes:
\begin{equation}
\tau_w = \frac{\mathbb{E}[w(X_i) \cdot Y_i(1)]}{\mathbb{E}[w(X_i)]} - \frac{\mathbb{E}[w(X_i) \cdot Y_i(0)]}{\mathbb{E}[w(X_i)]}.
\end{equation}
The weights $w(X_i)$ are designed to minimize the uncertainty of the estimate, often by upweighting units with more precise treatment effect estimates or units in regions of covariate space with better overlap between treatment groups. This approach trades off some bias (by changing the estimand) for substantial variance reduction, which can be particularly valuable in decision-making contexts where minimizing estimation uncertainty is paramount.

            \section{Proofs for Theorems in Section~\ref{sec:val}}
\textbf{\textit{Proof~Theorem~\ref{thm:unbiased_estimator}.}}
We need to show that $\mathbb{E}[\hat{\theta}_{MSE}(\alpha, \beta)] = \theta_{MSE}(\alpha, \beta)$, which is equivalent to showing that $\mathbb{E}[\hat{\epsilon}^{(s)}_\alpha] = \epsilon^{(s)}_\alpha + \epsilon^{(s)}_{DM}$ for any estimator $\alpha$ and study $s$, where the additional $\epsilon^{(s)}_{DM}$ term cancels out in the difference $\hat{\theta}_{MSE}(\alpha, \beta) = \hat{M}_\beta - \hat{M}_\alpha$.

For the MSE metric, our cross-fitted performance estimate uses:
\begin{equation}
\Delta_{MSE}(\hat{\tau}^{(s)}_{\alpha,-k}, \hat{\tau}^{(s)}_{\text{DM},k}) = (\hat{\tau}^{(s)}_{\alpha,-k} - \hat{\tau}^{(s)}_{\text{DM},k})^2.
\end{equation}

Expanding this squared difference by adding and subtracting the true treatment effect $\tau^{(s)}$:
\begin{align}
(\hat{\tau}^{(s)}_{\alpha,-k} - \hat{\tau}^{(s)}_{\text{DM},k})^2 &= (\hat{\tau}^{(s)}_{\alpha,-k} - \tau^{(s)} + \tau^{(s)} - \hat{\tau}^{(s)}_{\text{DM},k})^2 \\
&= (\hat{\tau}^{(s)}_{\alpha,-k} - \tau^{(s)})^2 + (\tau^{(s)} - \hat{\tau}^{(s)}_{\text{DM},k})^2 \\
&\quad + 2(\hat{\tau}^{(s)}_{\alpha,-k} - \tau^{(s)})(\tau^{(s)} - \hat{\tau}^{(s)}_{\text{DM},k}).
\end{align}

Taking conditional expectations given $S = s$:
\begin{align}
&\mathbb{E}[\Delta_{MSE}(\hat{\tau}^{(s)}_{\alpha,-k}, \hat{\tau}^{(s)}_{\text{DM},k}) \mid S = s] \\
&= \mathbb{E}[(\hat{\tau}^{(s)}_{\alpha,-k} - \tau^{(s)})^2 \mid S = s] + \mathbb{E}[(\tau^{(s)} - \hat{\tau}^{(s)}_{\text{DM},k})^2 \mid S = s] \\
&\quad + 2\mathbb{E}[(\hat{\tau}^{(s)}_{\alpha,-k} - \tau^{(s)})(\tau^{(s)} - \hat{\tau}^{(s)}_{\text{DM},k}) \mid S = s].
\end{align}

The first term equals $\mathbb{E}[\Delta_{MSE}(\hat{\tau}^{(s)}_{\alpha,-k}, \tau^{(s)}) \mid S = s]$, and the second term equals $\mathbb{E}[\Delta_{MSE}(\hat{\tau}^{(s)}_{\text{DM},k}, \tau^{(s)}) \mid S = s]$.

For the cross term, since our cross-fitting procedure ensures that $\hat{\tau}^{(s)}_{\alpha,-k}$ and $\hat{\tau}^{(s)}_{\text{DM},k}$ are computed on disjoint data subsets, they are conditionally independent given $S = s$ and $\tau^{(s)}$. Since $\hat{\tau}^{(s)}_{\text{DM},k}$ is unbiased ($\mathbb{E}[\hat{\tau}^{(s)}_{\text{DM},k} \mid S = s] = \tau^{(s)}$), we have:
\begin{align}
&\mathbb{E}[(\hat{\tau}^{(s)}_{\alpha,-k} - \tau^{(s)})(\tau^{(s)} - \hat{\tau}^{(s)}_{\text{DM},k}) \mid S = s] \\
&= \mathbb{E}[\hat{\tau}^{(s)}_{\alpha,-k} - \tau^{(s)} \mid S = s] \cdot \mathbb{E}[\tau^{(s)} - \hat{\tau}^{(s)}_{\text{DM},k} \mid S = s] = 0.
\end{align}

Therefore:
\begin{equation}
\mathbb{E}[\hat{\epsilon}^{(s)}_\alpha] = \frac{1}{K} \sum_{k=1}^K \mathbb{E}[\Delta_{MSE}(\hat{\tau}^{(s)}_{\alpha,-k}, \hat{\tau}^{(s)}_{\text{DM},k}) \mid S = s] = \epsilon^{(s)}_\alpha + \epsilon^{(s)}_{DM}.
\end{equation}

Similarly, $\mathbb{E}[\hat{\epsilon}^{(s)}_\beta] = \epsilon^{(s)}_\beta + \epsilon^{(s)}_{DM}$.

Taking the difference:
\begin{equation}
\mathbb{E}[\hat{\epsilon}^{(s)}_\beta - \hat{\epsilon}^{(s)}_\alpha] = (\epsilon^{(s)}_\beta + \epsilon^{(s)}_{DM}) - (\epsilon^{(s)}_\alpha + \epsilon^{(s)}_{DM}) = \epsilon^{(s)}_\beta - \epsilon^{(s)}_\alpha.
\end{equation}

Finally:
\begin{align}
\mathbb{E}[\hat{\theta}_{MSE}(\alpha, \beta)] &= \mathbb{E}\left[\frac{1}{N} \sum_{s=1}^N (\hat{\epsilon}^{(s)}_\beta - \hat{\epsilon}^{(s)}_\alpha)\right] \\
&= \frac{1}{N} \sum_{s=1}^N (\epsilon^{(s)}_\beta - \epsilon^{(s)}_\alpha) = M_\beta - M_\alpha = \theta_{MSE}(\alpha, \beta).
\end{align}

\textbf{\textit{Proof of Theorem~\ref{thm:regret_bias_bounds}.}}
For any estimator $\gamma$ and study $s$, define:
\begin{align}
R^{(s)}_{\gamma,\text{true}} &= |\mathcal{O}_s| \times \left| \tau^{(s)} \times \left(\mathbf{1}[\tau^{(s)} > 0] - \mathbf{1}[\hat{\tau}^{(s)}_{\gamma,-k} > c \cdot \text{SE}(\hat{\tau}^{(s)}_{\gamma,-k})]\right) \right|, \\
R^{(s)}_{\gamma,\text{est}} &= |\mathcal{O}_s| \times \left| \hat{\tau}^{(s)}_{\text{DM},k} \times \left(\mathbf{1}[\hat{\tau}^{(s)}_{\text{DM},k} > 0] - \mathbf{1}[\hat{\tau}^{(s)}_{\gamma,-k} > c \cdot \text{SE}(\hat{\tau}^{(s)}_{\gamma,-k})]\right) \right|.
\end{align}

From our triangle inequality analysis, we have:
\begin{equation}
|R^{(s)}_{\gamma,\text{est}} - R^{(s)}_{\gamma,\text{true}}| \leq |\mathcal{O}_s| \times \left(|\tau^{(s)}| \times |\mathbf{1}[\tau^{(s)} > 0] - \mathbf{1}[\hat{\tau}^{(s)}_{\text{DM},k} > 0]| + |\tau^{(s)} - \hat{\tau}^{(s)}_{\text{DM},k}|\right).
\end{equation}

Taking expectations conditional on $S = s$ and averaging over folds:
\begin{equation}
|\hat{\epsilon}^{(s)}_\gamma - \epsilon^{(s)}_\gamma| \leq B^{(s)}_\gamma.
\end{equation}

For the overall performance estimates:
\begin{align}
|\hat{M}_\gamma - M_\gamma| &= \left|\frac{1}{N} \sum_{s=1}^N (\hat{\epsilon}^{(s)}_\gamma - \epsilon^{(s)}_\gamma)\right| \\
&\leq \frac{1}{N} \sum_{s=1}^N |\hat{\epsilon}^{(s)}_\gamma - \epsilon^{(s)}_\gamma| \\
&\leq \frac{1}{N} \sum_{s=1}^N B^{(s)}_\gamma.
\end{align}

Finally, for the test statistic:
\begin{align}
|\hat{\theta}_{\text{regret}}(\alpha, \beta) - \theta_{\text{regret}}(\alpha, \beta)| &= |(\hat{M}_\beta - M_\beta) - (\hat{M}_\alpha - M_\alpha)| \\
&\leq |\hat{M}_\beta - M_\beta| + |\hat{M}_\alpha - M_\alpha| \\
&\leq \frac{1}{N} \sum_{s=1}^N (B^{(s)}_\beta + B^{(s)}_\alpha).
\end{align}

\textbf{\textit{Proof for Theorem~\ref{thm:regret_expected_bias}.}}
For any estimator $\gamma$ and study $s$, the difference between estimated and true regret is:
\begin{align}
&\Delta_{\text{Regret}}(\hat{\tau}^{(s)}_{\gamma,-k}, \hat{\tau}^{(s)}_{\text{DM},k}, c) - \Delta_{\text{Regret}}(\hat{\tau}^{(s)}_{\gamma,-k}, \tau^{(s)}, c) \\
&= |\mathcal{O}_s| \times \left[(\hat{\tau}^{(s)}_{\text{DM},k} - \tau^{(s)}) \times \left(\mathbf{1}[\hat{\tau}^{(s)}_{\text{DM},k} > 0] - \mathbf{1}[\hat{\tau}^{(s)}_{\gamma,-k} > c \cdot \text{SE}]\right)\right. \\
&\quad \left.+ \tau^{(s)} \times \left(\mathbf{1}[\hat{\tau}^{(s)}_{\text{DM},k} > 0] - \mathbf{1}[\tau^{(s)} > 0]\right)\right].
\end{align}

Taking expectations conditional on $S = s$:
\begin{equation}
\mathbb{E}[\hat{\epsilon}^{(s)}_\gamma - \epsilon^{(s)}_\gamma] = B^{(s)}_\gamma.
\end{equation}

For the sign of the second term:
\begin{itemize}
\item If $\tau^{(s)} > 0$: The term becomes $\tau^{(s)} \times (P(\hat{\tau}^{(s)}_{\text{DM},k} > 0 \mid S = s) - 1) \leq 0$
\item If $\tau^{(s)} < 0$: The term becomes $\tau^{(s)} \times P(\hat{\tau}^{(s)}_{\text{DM},k} > 0 \mid S = s) \leq 0$
\item If $\tau^{(s)} = 0$: The term equals $0$
\end{itemize}

The overall bias follows from linearity of expectation.

\begin{corollary}[Bias Direction Under Sign Consistency]
If the DM estimator has perfect sign consistency ($P(\text{sign}(\hat{\tau}^{(s)}_{\text{DM},k}) = \text{sign}(\tau^{(s)}) \mid S = s) = 1$), then:
\begin{equation}
B^{(s)}_\gamma = |\mathcal{O}_s| \times \mathbb{E}\left[(\hat{\tau}^{(s)}_{\text{DM},k} - \tau^{(s)}) \times \left(\mathbf{1}[\hat{\tau}^{(s)}_{\text{DM},k} > 0] - \mathbf{1}[\hat{\tau}^{(s)}_{\gamma,-k} > c \cdot \text{SE}]\right) \mid S = s\right].
\end{equation}
\end{corollary}

This result reveals several important insights:
\begin{enumerate}
\item Sign errors create systematic bias: The second term $\tau^{(s)} \times (P(\hat{\tau}^{(s)}_{\text{DM},k} > 0 \mid S = s) - \mathbf{1}[\tau^{(s)} > 0])$ is always non-positive, creating a conservative bias that underestimates regret.

\item The bias is estimator-specific: Different estimators $\alpha$ and $\beta$ will have different bias terms $B^{(s)}_\alpha$ and $B^{(s)}_\beta$, but these may partially cancel in the test statistic $\hat{\theta}_{\text{regret}}(\alpha, \beta)$.

\item Magnitude of bias depends on:
   \begin{itemize}
   \item Accuracy of the DM estimator ($|\hat{\tau}^{(s)}_{\text{DM},k} - \tau^{(s)}|$)
   \item Sign consistency of the DM estimator
   \item Rollout probabilities of the estimators being compared
   \end{itemize}

\item Asymptotic unbiasedness: As DM estimator precision improves, both terms vanish, confirming asymptotic unbiasedness.
\end{enumerate}
        \end{appendix}

\end{document}